\documentclass{IEEEtran}
\usepackage{graphicx}
\usepackage{subfigure}
\usepackage{amsmath}
\usepackage{amsthm}
\usepackage{amssymb}
\usepackage{mathrsfs}
\usepackage{slashbox}
\usepackage{float}
\usepackage[vlined,ruled,linesnumbered]{algorithm2e}
\usepackage{xcolor}
\usepackage{soul}
\usepackage{multirow}

\newtheorem{thm}{Theorem}%[section]

\newtheorem{lem}{Lemma}

\newtheorem{definition}{Definition}
\newtheorem{property}{Property}

\newcount\DraftStatus  % 0 suppresses notes to selves in text
\usepackage{color}
\definecolor{darkgreen}{rgb}{0,0.5,0}
\definecolor{purple}{rgb}{1,0,1}
\newcommand{\draftnote}[2]{\ifnum\DraftStatus=1s
	\marginpar{
		\tiny\raggedright
		\hbadness=10000
        \def\baselinestretch{0.8}
        \textcolor{#1}{\textsf{\hspace{0pt}#2}}}
     \fi}
\begin{document}
\title{One Parameter Defense - Defending against Data Inference Attacks via Differential Privacy}


\author{Dayong Ye, Sheng Shen, Tianqing Zhu*, Bo Liu and Wanlei Zhou% <-this % stops a space
\thanks{*Tianqing Zhu is the corresponding author. D. Ye, S. Shen, T. Zhu and B. Liu are with the Centre for Cyber Security and Privacy and the School of Computer Science, University of Technology, Sydney, Australia, 2007. Wanlei Zhou is with the City University of Macau, Macau, PR China. Email: \{Dayong.Ye, 12086892, Tianqing.Zhu, Bo.Liu\}@uts.edu.au, wlzhou@cityu.mo.}}% <-this % stops a space

\maketitle

\begin{abstract}
Machine learning models are vulnerable to data inference attacks, 
such as membership inference and model inversion attacks. 
In these types of breaches, an adversary attempts to infer a data record's membership in a dataset 
or even reconstruct this data record using a confidence score vector predicted by the target model. 
However, most existing defense methods only protect against membership inference attacks. 
Methods that can combat both types of attacks require a new model to be trained, 
which may not be time-efficient. 
In this paper, we propose a differentially private defense method 
that handles both types of attacks in a time-efficient manner 
by tuning only one parameter, the privacy budget. 
The central idea is to modify and normalize the confidence score vectors 
with a differential privacy mechanism which preserves privacy and 
obscures membership and reconstructed data. 
Moreover, this method can guarantee the order of scores in the vector  
to avoid any loss in classification accuracy. 
The experimental results show the method to be an effective and timely defense against 
both membership inference and model inversion attacks with no reduction in accuracy.
\end{abstract}

\section{Introduction}
On the back of the massive amounts of data we humans generate every day, 
machine learning (ML) has become a key part of many real-world applications, 
ranging from image classification to speech recognition \cite{Salem19}. 
%Machine learning has also been provided as an online service by many platforms, 
%e.g., machine-learning-as-a-service from Amazon \cite{Amazon}. 
%The success of machine learning is mainly due to the availability of a massive amount of data. 
However, these data often contain sensitive personal information 
that is vulnerable to a range of adversarial activities, including 
membership inference attacks \cite{Shokri17,Nasr19} and model inversion attacks \cite{Fred15,Yang19}. 
Both fall into the category of data inference attacks, 
which are launched by exploiting redundant information contained in confidence score vectors. 
For example, a machine learning model will usually be more confident in its prediction about 
a data record that is in its training dataset, 
over another data record that is not. 
Attackers can exploit this difference in confidence to determine 
whether a given data record is or not a member of the target model's training dataset. 
Here, a confidence score vector is a probability distribution over the possible classes predicted by an ML model. 
Each score in the vector indicates the model's confidence in a prediction of the corresponding class. 
The class with the largest confidence is predicted as the label of the input data record. 

Data inference attacks can result in severe privacy violations. 
For example, consider a model that has been trained on data 
collected from people with a certain disease. 
If a particular individual's data are known to be in the training dataset, 
the adversary can immediately infer that person's health status. 
Another example is model inversion attack. 
As shown in Fig. \ref{fig:Inversion}, an adversary can train an attack model to accurately reconstruct 
an input data record using only the confidence score vector, 
even if the data record is never seen by the attack model. 
Hence, defending against data inference attacks has been 
the focus of much attention in the privacy community. 
\begin{figure}[ht]
\centering
	\includegraphics[scale=0.45]{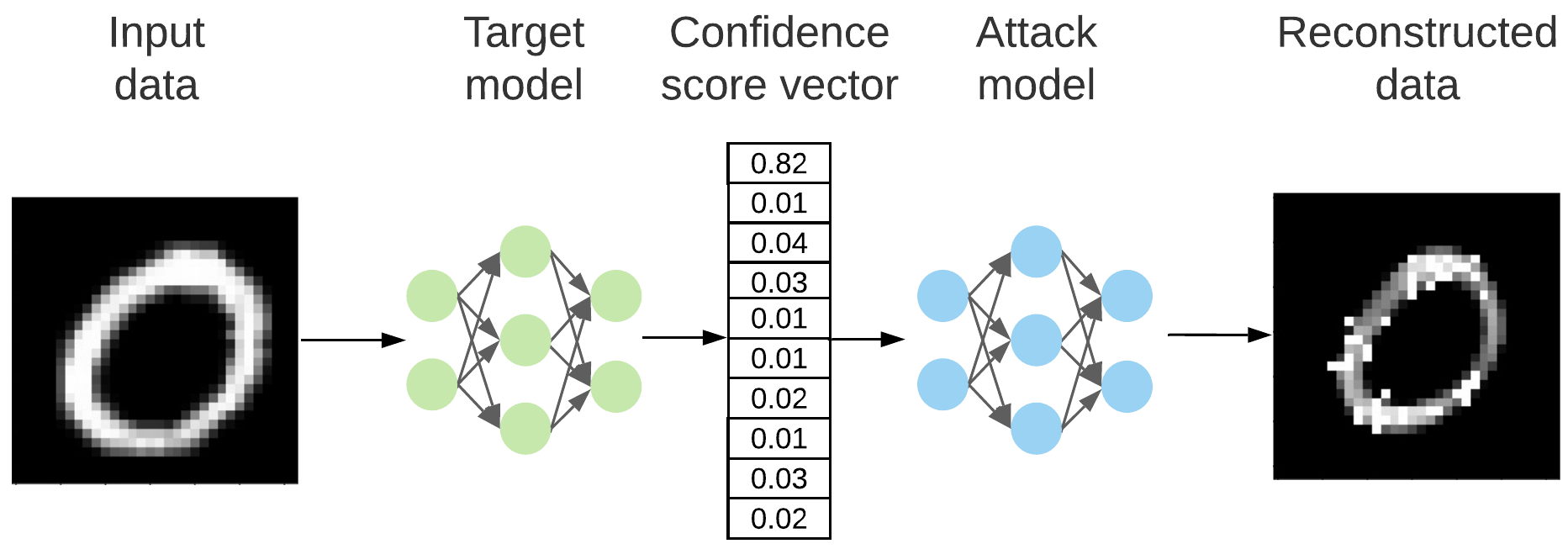}
	\caption{An input data record is accurately reconstructed 
	by an attack model which has never seen this data record.}
	\label{fig:Inversion}
\end{figure}

The methods proposed to date can be roughly classified into four categories 
based on the defense techniques employed. 
Those in the first category use regularization techniques to reduce overfitting, 
such as $L_2$ regularizer \cite{Shokri17}, dropout \cite{Salem19}, 
model-stacking \cite{Salem19} and min-max regularization \cite{Nasr18}. 
This is because overfitting is one of the major factors leading to the distinguishability 
between member and non-member data records \cite{Shokri17}. 
The shortcoming of these methods, however, is that distinguishability, 
i.e., a model's vulnerability, is not reduced directly. 
Moreover, they require retraining the target model 
which may not be very efficient for complex neural networks. 

The second category of the methods is based on adversarial examples \cite{Jia19}. 
These methods add a carefully crafted noise vector to a confidence score vector 
turning it into an adversarial example to mislead the attacker's classifier. 
These methods have formal utility loss guarantees of confidence score vectors. 
However, their effectiveness depends on the transferability of adversarial examples \cite{Papernot16}, 
which might not generally reduce the distinguishability of prediction scores \cite{Yang20}. 

The third category of the methods is based on deep neural networks \cite{Yang20}. 
These methods interpret defense goals as loss functions 
and train the deep neural networks as defense models to defend against attacks. 
These methods can minimize the content of information 
attackers can use to infer membership or reconstruct the data records. 
However, this strategy requires training new models 
which may not be very time-efficient. 

The fourth category describes the differential privacy methods \cite{Abadi16}. 
These methods carry a theoretical guarantee of privacy. 
However, as differential privacy is usually used during the training process, 
e.g., adding noise to gradients, 
there is always a large classification accuracy loss \cite{Nasr18}. 
As shown by Jayaraman and Evans \cite{Jaya19},  
existing differentially private machine learning methods 
rarely offer acceptable privacy-utility tradeoffs for complex models. 
Moreover, since existing methods have to be integrated into the training of target models, 
they are not applicable to those target models which have already been deployed.

A common limitation shared by the existing defense methods is that all but one only protect against membership inference attacks. 
The only method to handle both membership inference and model inversion attacks was proposed by Yang et al. \cite{Yang20}. 
Their method, however, requires training new models which may not be very time-efficient. 
In this paper, we develop a time-efficient defense method against both membership inference and model inversion attacks. 
%This method can simultaneously guarantee the mitigation of attacks, classification accuracy, 
%utility-loss of confidence score vectors and time-efficiency.
%However, simultaneously guaranteeing the four aspects is challenging, 
%because in some situations they may contradict each other, 
%e.g., effective mitigation of attacks versus time-efficiency.
%To address this challenge, 
Our solution is a differential privacy mechanism that modifies and normalizes the confidence score vectors to confuse the attacker's classifier. 
As such, the only parameter that needs to be tuned is the privacy budget, 
which controls the amount of perturbation added to the vector. 
For example, when a learning model is very confident in predicting a given data record, 
the output confidence score vector will have a very high probability of predicting one class 
and very low probabilities of predicting the others. 
By using our method to modify the confidence score vector, 
the probability distribution can be reshaped so that 
one class has a slightly higher probability for prediction than other classes. 
This new confidence score vector appears less confident than the original one 
and can be used to confuse the attacker's classifier. 

The idea of modifying confidence score vectors to defend against data inference attacks is not new. 
However, the methods based on this premise, such as those in \cite{Shokri17,Jia19}, do not provide a privacy guarantee 
plus they can only defend against membership inference attacks. 
By comparison, our method has a privacy guarantee, and offers protection against 
both membership inference and model inversion attacks. 
%Specifically, our method adopts differential privacy, 
%so that the privacy of confidence score vectors can be preserved 
%and the attacks can be mitigated. 
Further, our method preserves the order of scores in confidence score vectors, 
which guarantees zero classification accuracy loss. 
Moreover, as our method does not require training any new models, 
time-efficiency can be achieved. 
In summary, this paper makes the following contributions. 
\begin{itemize}
    \item We are the first to propose a one-parameter defense method that requires only one parameter to be tuned, the privacy budget. 
    This method guarantees both differential privacy and time-efficiency against both membership inference and model inversion attacks. 
    \item We theoretically demonstrate how to tune the privacy budget 
    to defend against both types of attacks, while 
    controlling the utility loss of confidence score vectors.
    \item We empirically show that 
    the presented method effectively mitigates both types of attacks 
    with no loss of classification accuracy, zero training time, and very low test time.
\end{itemize}

The rest of this paper is organized as follows. 
Section \ref{sec:related work} reviews the existing defense methods. 
Section \ref{sec:preliminaries} presents the preliminaries regarding 
data inference attacks and differential privacy. 
Our method is presented in detail in Section \ref{sec:method}. 
%followed by the theoretical analysis in Section \ref{sec:analysis}. 
Section \ref{sec:analysis} theoretically analyzes the properties of our method and demonstrates how to use the method to defend against data inference attacks. 
The experimental results are provided in Section \ref{sec:experiments}. 
Finally, the paper concludes in Section \ref{sec:conclusion}.

\section{Related work}\label{sec:related work}
Defense methods can be roughly classified into four categories based on the adopted techniques: 
regularization, adversarial examples, deep neural networks and differential privacy. 
%We first review these methods and then discuss their limitations.

\subsection{Regularization-based defense methods}
Shokri et al. \cite{Shokri17} investigated four methods of defense 
against membership inference attacks. 
The first method is to restrict the prediction vector to the top $k$ classes, 
where a smaller $k$ means less information is leaked. 
The second method is to coarsen the precision of the confidence score vector 
by rounding the classification probabilities within the vector down to $d$ floating point digits. 
Again, a smaller $d$ means less information is leaked. 
The third method is to increase the entropy of the confidence score vector 
via a softmax function with a temperature $t$ to compute the output of the logits vector. 
The fourth method is to use $L_2$-norm standard regularization.

Nasr et al. \cite{Nasr18} formalized the interactions 
between their defense method and a membership inference attack as a min-max privacy game. 
To find the solution to the game, they train a target model using an adversarial process 
that minimizes both the prediction loss of the model and 
the maximum gain of the inference attacks. 
Through this approach, the target model provides both membership privacy 
and strong regularization capability.

Salem et al. \cite{Salem19} proposed two defense methods against membership inference attacks. 
The first method is called `dropout' 
which randomly deletes a fixed proportion of edges 
from a fully connected neural network model in each training iteration to avoid overfitting. 
The second method is called `model stacking', 
which is based on ensemble learning and 
constructs the target model using three different machine learning models. 
Two models are placed in the first layer to take the original training data 
while the third is trained with the conference score vectors of the first two models.
The idea of model stacking is to arrange multiple models in a hierarchy so as to avoid overfitting.

\subsection{Adversarial example-based defense methods}
Jia et al. \cite{Jia19} proposed a defense method named MemGuard 
that adds noise to each confidence score vector to make it an adversarial example. 
Their idea is based on the fact that 
deep learning models can be misled by adversarial examples to produce wrong predictions \cite{Goodfellow15,Papernot16SP}. 
They formalized the process of adding noise as an optimization problem 
and developed an algorithm to solve the problem based on gradient descent. 

\subsection{Deep neural network-based defense methods}
Yang et al. \cite{Yang20} designed a purifier model 
that takes a confidence score vector as input and 
reshapes it to meet defense goals. 
The purifier model consists of an encoder and a decoder. 
The encoder maps the confidence score vector predicted by the target model 
to a latent representation. 
The decoder then maps the latent representation to a reconstruction of the confidence score vector.

\subsection{Differential privacy-based defense methods} 
Differential privacy has been a prevalent tool to preserve the privacy of deep learning models \cite{Shokri15,Abadi16}. 
A comprehensive survey regarding differential privacy in deep learning can be found in \cite{Jaya19,Gong20}. 
To implement differentially private deep learning, 
noise can be added to one of the five places in a deep neural network: 
input datasets \cite{Heikkila17}, loss functions \cite{Zhao19}, gradients \cite{Abadi16,Cheng18}, 
weights of neural networks \cite{Jaya18,Phan19}, and output classes \cite{Papernot17,Papernot18}.

Heikkila et al. \cite{Heikkila17} proposed a general approach for 
a privacy-preserving learning schema for distributed settings. 
Their approach combines secure multiparty communication 
with differentially private Bayesian learning methods. 
%so as to achieve distributed differentially private Bayesian learning. 
In their approach, each client adds a Gaussian noise to the data and divides the noised data into shares. 
%and divides them and the noise into shares. 
The shares are not divided independently. 
Instead, they are divided using a fixed-point representation of real numbers which allows exact cancellation of the noise in the sum.
Each share is then sent to a server. 
This way, the sum of the shares discloses the real value, 
but separately they are just random noise.

Zhao et al. \cite{Zhao19} proposed a privacy-preserving collaborative deep learning system. 
The system allows users to collaboratively build a collective learning model 
while only sharing the parameters, not the data. 
To preserve the private information embodied in the parameters, 
they developed a functional mechanism, an extended version of the Laplace mechanism, 
to perturb the objective function of the neural network.

Cheng et al. \cite{Cheng18} developed a privacy-preserving algorithm for distributed learning  
based on a leader-follower framework, 
where the leaders guide the followers in the right direction to improve their learning speed. 
For efficiency, communication is limited to leader-follower pairs. 
To preserve the privacy of the leaders, 
Gaussian noise is added to the gradients of the leaders' learning models.

Phan et al. \cite{Phan19} proposed a heterogeneous Gaussian mechanism 
to preserve privacy in deep neural networks. 
Unlike a regular Gaussian mechanism, 
this heterogeneous Gaussian mechanism can arbitrarily redistribute noise 
from the first hidden layer and the gradient of the model 
to achieve an ideal trade-off between model utility and privacy loss. 
To obtain the property of arbitrary redistribution, 
they introduce a noise redistribution vector that can be used  
to change the variance of the Gaussian distribution. 
Further, it can be guaranteed that, 
by adapting the values of the scores in the noise redistribution vector, 
more noise can be added to the more vulnerable components of the model to improve robustness and flexibility.

Papernot et al. \cite{Papernot17} developed a model called Private Aggregation of Teacher Ensembles (PATE)
which has been successfully applied to generative adversarial nets (GANs) for a privacy guarantee  \cite{Jordon19}. 
PATE consists an ensemble of $n$ teacher models; 
an aggregation mechanism; and a student model. 
Each teacher model is trained independently on a subset of private data.
%The aggregation mechanism is used by teachers to collectively make a decision, by vote regarding the output class of a given sample. 
To protect the privacy of the data labels, Laplace noise is added to the output classes, i.e., the teacher votes.
Last, the student model is trained through knowledge transfer from the teacher ensemble with the public data and privacy-preserving labels. 
Later, Papernot et al. \cite{Papernot18} improved the PATE model to make it applicable to large-scale tasks and real-world datasets. 
%They first replace the Laplace noise added to the teacher votes with Gaussian noise, 
%which is better suited to data-dependent privacy analysis. 
%They next modify the aggregation mechanism by introducing confident and interactive aggregators 
%that select queries worth answering in a privacy-preserving way.

In addition to the utilization of standard differential privacy, local differential privacy has also been adopted recently. For example, Kim et al. \cite{Kim21} adopted Gaussian mechanism to preserve local differential privacy of user data in federated learning models. They also analyzed the trade-offs between user privacy, global utility and transmission rate, where a larger noise variance guarantees a stronger privacy with a lower utility bound and a higher transmission rate bound. 

\subsection{Discussion of related work}
Among all these defense methods, the regularization and 
adversarial example methods are only designed to defend against membership inference attacks. 
The differential privacy methods have mostly been developed for general privacy preservation 
rather than to protect against a specific inference attack. 
Only the deep neural network-based method \cite{Yang20} has been specifically developed to defend 
against both membership inference and model inversion attacks. 
That method, however, does require new deep neural networks to be trained. 
Hence, it is not an optimal strategy if time efficiency is a concern. 

Our strategy of modifying then normalizing the output confidence score vectors 
with a differential privacy mechanism provides protection against 
both membership inference and model inversion attacks. 
Unlike the deep neural network-based method in \cite{Yang20}, 
no new models need to be trained, so our method is very time-efficient. 
Plus, the order of scores in the confidence vectors is preserved, 
which skirts the accuracy tradeoff inherent to other differential privacy methods. 
%Moreover, by properly tuning the value of privacy budget, 
%the useful information contained a confidence score vector can be as close as expected to its modified version. 
%Thus, both the classification accuracy and 
%utility loss of confidence score vectors can be finely controlled.

\section{Preliminaries}\label{sec:preliminaries}
\subsection{Data inference attacks}
Data inference attacks can come in the form of either a membership inference or a model inversion attack, each of which has a different inference goal. 
In both of these attacks, the target model is evaluated by its owner but is open to public users, including attackers. The model can be accessed only in a black-box manner, where an attacker inputs a data sample into the model and receives the corresponding output. The real output confidence scores are modified by the owner while only the modified scores are published to the attacker. The modification process is integrated into the model and thus, is hidden to the attacker. Hence, the attacker cannot access any intermediate classification values, and he can only exploit the modified scores to invert the model or infer the membership of samples. Moreover, the structure and parameters of the target model are also unknown to the attacker. The attacker is aware of the distribution of the training dataset of the target model and can collect a new dataset based on the same distribution. The attacker, however, cannot directly query the training dataset of the target model. These assumptions regarding the attacker are the same as those made in \cite{Yang19}.

\subsubsection{Membership inference attacks}
In a membership inference attack, the goal is to infer whether a given data record is in the training dataset of the target model \cite{Nasr19}. 
%Membership inference attacks can be further classified into two categories: confidence-based attacks \cite{Shokri17} and confidence and label-based attacks \cite{Nasr18}. 
With these types of attacks, the attacker has access to the model's prediction scores. 
Thus, the typical strategy for inferring the record's membership is to train an attack model which is usually a binary classifier. 
The input is either the confidence score vector of the given data record, or its label, or both. 
Then, the output is a prediction over the record's membership in the training dataset. 
Before training the attack model, the attacker has to train a shadow model on an 
auxiliary dataset drawn from the same data distribution as the training data 
of the target model as a way of replicating it \cite{Salem19}. 
The attack model is then trained on the confidence score vectors 
of membership/non-membership predicted by the shadow model.

Formally, let $F$ be the target model 
and $D^{\rm train}_{\rm target}$ be the private training dataset of the target model. 
$D^{\rm train}_{\rm target}$ contains a set of labeled data records $(\mathbf{x}^{(i)},y^{(i)})$, 
where $i$ denotes the index of the data record, 
$\mathbf{x}^{(i)}$ represents the data, and $y^{(i)}$ is the label of $\mathbf{x}^{(i)}$. 
The output of the target model $F(\mathbf{x}^{(i)})$ is a vectorized class of the data $\mathbf{x}^{(i)}$. 
Additionally, let $S$ be the shadow model and 
$D^{\rm train}_{\rm shadow}$ be the auxiliary training dataset of the shadow model 
which has the same data distribution as $D^{\rm train}_{\rm target}$, 
and let $G$ be the attack model. 
Given a data record $(\mathbf{x},y)$, the attacker inputs $\mathbf{x}$ into the target model $F$ 
and receives an output $F(\mathbf{x})$, 
and then feeds $F(\mathbf{x})$ into the attack model $G$ 
and receives an output $G(F(\mathbf{x}))$ 
which is the membership probability $Pr[(\mathbf{x},y)\in D^{\rm train}_{\rm target}]$ 
that the data record $(\mathbf{x},y)$ belongs to the ``\emph{in}'' class.

\subsubsection{Model inversion attacks}
Model inversion attacks aim to reconstruct the input data from the confidence score vectors predicted by the target model \cite{Fred15}. 
The attacker trains a separate attack model on an auxiliary dataset 
which acts as the inverse of the target model \cite{Yang19}. 
The attack model takes the confidence score vectors of the target model as input 
and tries to output the original input data of the target model.

Formally, let $F$, again, be the target model and $G$ be the attack model. 
Given a data record $(\mathbf{x},y)$, the attacker inputs $\mathbf{x}$ into $F$ and receives $F(\mathbf{x})$, 
and then feeds $F(\mathbf{x})$ into $G$ and receives $G(F(\mathbf{x}))$ 
which is expected to be very similar to $\mathbf{x}$, i.e., $G(F(\mathbf{x}))\approx \mathbf{x}$.

\subsection{Differential privacy}\label{sub:DP}
Differential privacy is a prevalent privacy model capable of guaranteeing that
any individual record being stored in or removed from a dataset makes little difference to the analytical output of the dataset \cite{Dwork14}. 
Differential privacy has been broadly applied in many research areas, 
including artificial intelligence \cite{Zhu20}, multi-agent systems \cite{Ye19,Ye20b,Ye20} and cyber security \cite{Ye21}. 
The concept is best described as follows. 
Two datasets $D$ and $D'$ are neighboring datasets if they differ by only one record. 
A query $f$ is a function that maps dataset $D$ to a range $\mathbb{R}$, $f: D\rightarrow\mathbb{R}$. 
%A group of queries is denoted as $F=\{f_{1},...,f_{m}\}$, and $F(D)$ denotes $\{f_{1}(D),...,f_{m}(D)\}$.
The maximal difference between the results of query $f$ on $D$ and $D'$ 
is defined as the sensitivity $\Delta S$,
which determines how much perturbation is required for a privacy-preserving answer.
The formal definition of differential privacy is given as follows.

\begin{definition}[$\epsilon$-Differential Privacy \cite{Dwork14}]\label{Def-DP}
A mechanism $\mathcal{M}$ provides $\epsilon$-differential privacy for any pair of neighboring datasets $D$ and $D'$, and for every set of outcomes $\Omega$, if $\mathcal{M}$ satisfies:
\begin{equation}
Pr[\mathcal{M}(D) \in \Omega] \leq \exp(\epsilon) \cdot Pr[\mathcal{M}(D') \in \Omega]
\end{equation}
\end{definition}

\begin{definition}[Sensitivity \cite{Dwork14}]\label{Def-GS} For a query $f:D\rightarrow\mathbb{R}$, the sensitivity of $f$ is defined as
\begin{equation}
\Delta S=\max_{D,D'} ||f(D)-f(D')||_{1}
\end{equation}
\end{definition}

\begin{definition}[Private prediction interface \cite{Dwork18}]\label{def:PPI}
A prediction interface $\mathcal{M}$ is $\epsilon$-differentially private, if for every interactive query generating algorithm $Q$, the output $(Q\rightleftarrows\mathcal{M}(F))$ is $\epsilon$-differentially private with respect to model $F$, where $(Q\rightleftarrows\mathcal{M}(F))$ denotes the sequence of queries and responses generated in the interaction of $Q$ and $\mathcal{M}$ on model $F$.
\end{definition}

This definition shows that a private prediction interface can guarantee the privacy preservation of the interaction between queries and responses. This definition meets the property of our problem, where an attacker interacts with a prediction interface, i.e., our defense method, on a target model. If this interaction is differentially private, for each query of the attacker to the target model, he receives only an obfuscated response which breaks the relationship between the attacker's query and the corresponding response.

%\begin{definition}[Private prediction \cite{Dwork18}]\label{def:private prediction}
%\textcolor{blue}{Let $\mathcal{M}$ be a method which takes a dataset $D\in(X\times Y)^n$ and a data record $\mathbf{x}$ as input, 
%and outputs a value in $Y$, i.e., $\mathcal{M}(D,\mathbf{x})=\mathbf{y}$. 
%We say $\mathcal{M}$ is an $\epsilon$-differentially private prediction method, 
%if for every $\mathbf{x}\in X$, the output of $\mathcal{M}(D,\mathbf{x})$, i.e., $\mathbf{y}$, is differetially private with respect to $D$.}
%\end{definition}

%\textcolor{blue}{This definition states that for every $\mathbf{x}\in X$, a private prediction method can guarantee that the corresponding output $\mathbf{y}$ satisfies differential privacy. This property meets the requirement of our method that an output confidence score vector $\mathbf{y}$ should satisfy differential privacy so that an attacker cannot infer the membership of the input $\mathbf{x}$ or reconstruct $\mathbf{x}$ by exploiting $\mathbf{y}$. This can be explained by the fact that 1) as the privacy of output $\mathbf{y}$ is preserved with respect to dataset $D$, the attacker cannot tell whether the input $\mathbf{x}$ is in $D$; 2) the privacy of $\mathbf{y}$ is achieved by modifying its contents using a differential privacy mechanism which destroys the relationship between $\mathbf{x}$ and $\mathbf{y}$, thus, $\mathbf{x}$ cannot be reconstructed from $\mathbf{y}$.}

One of the most widely used differential privacy mechanisms is the exponential mechanism \cite{Zhu17} 
which defines a complex distribution over a large arbitrary domain.
\begin{definition}[The Exponential Mechanism \cite{Dwork14}]\label{Def-Ex}
The exponential mechanism $\mathcal{M}_E$ selects and outputs a score $r\in\mathcal{R}$
with probability proportional to $exp(\frac{\epsilon u(D,r)}{2\Delta u})$,
where $\epsilon$ is the privacy budget, $u(D,r)$ is the utility of a dataset and output pair,
and $\Delta u=\max\limits_{r\in\mathcal{R}} \max\limits_{D,D':||D-D'||_1\leq 1} |u(D,r)-u(D',r)|$
is the sensitivity of the utility score.
\end{definition}

\begin{figure*}[ht]
\centering
	\includegraphics[scale=0.5]{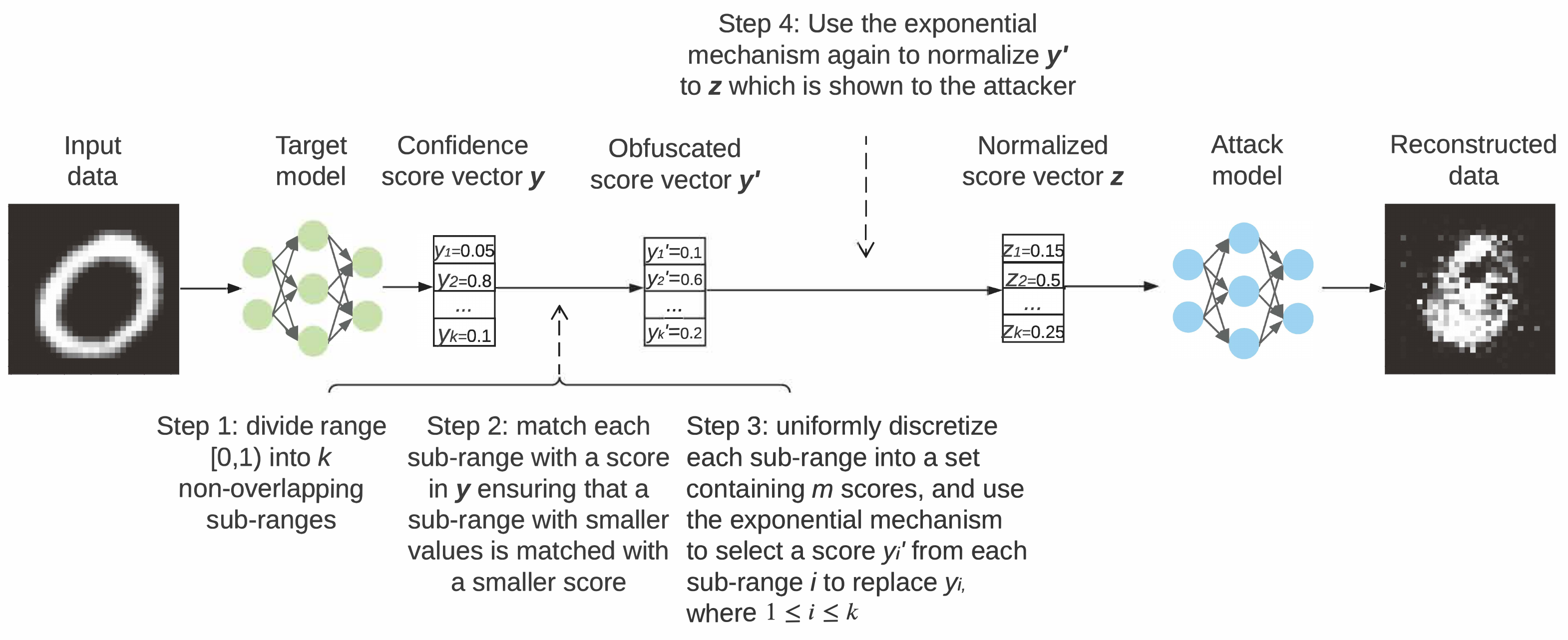}
	\caption{Overview of our method. The method consists of four steps. In Step 1, the range $[0,1)$ is divided into $k$ non-overlapping sub-ranges, where $k$ is the number of classes that the target model $T$ can classify. In Step 2, each sub-range is matched with a score in vector $\mathbf{y}$, where $\mathbf{y}$ is output by $T$, ensuring that a sub-range with smaller values is matched with a smaller score. In Step 3, each sub-range is uniformly discretized into a set with $m$ scores, where $m$ is a hyper-parameter. Then, for each sub-range $i$, the exponential mechanism is used to select a score from $m$ scores. The selected score is named $y'_i$ used to replace $y_i$. In Step 4, since it is likely that $\sum^k_{i=1}y'_i\neq 1$, the exponential mechanism is used again to normalize vector $\mathbf{y}'$ to $\mathbf{z}$ which is shown to the attacker.}
	\label{fig:example}
\end{figure*}

\section{The differentially private defense method}\label{sec:method}
\subsection{Overview of our method}
Our defense method, denoted as $\mathcal{M}_d$, is based on an exponential mechanism 
that perturbs the confidence score vector output of the target model $F(x)$. 
Let $\mathbf{y}=F(x)$ consist of $k$ scores: $\mathbf{y}=\langle y_1,...,y_k\rangle$. 
Fig. \ref{fig:example} displays the overview of our method which consists of two phases. 
The first phase involves three steps which focus on using an exponential mechanism to modify $\mathbf{y}$ to $\mathbf{y}'$. 
The second phase has one step to normalize $\mathbf{y}'$ to a valid confidence score vector $\mathbf{z}$ using the exponential mechanism again.

In the first phase, the range $[0,1)$ is divided into $k$ non-overlapping sub-ranges. 
This division is based on the original scores in vector $\mathbf{y}$ to ensure that 
each sub-range covers an original score and also, a sub-range with smaller values covers a smaller score. 
Specifically, suppose the scores in $\mathbf{y}$ have been ranked as $y_1\leq...\leq y_k$. 
Then, each sub-range $i$ is $[\frac{y_{i-1}+y_i}{2},\frac{y_i+y_{i+1}}{2})$, where $2\leq i\leq k-1$. 
For $i=1$, the sub-range is $[0,\frac{y_1+y_2}{2})$, and for $i=k$, the sub-range is $[\frac{y_{k-1}+y_k}{2},1)$.
%In the first phase, the range $[0,1)$ is equally divided into $k$ non-overlapping sub-ranges. 
%Each sub-range is matched with a score in $\mathbf{y}$, where a sub-range with smaller values is matched with a smaller score. 
%Each sub-range covers one score in $\mathbf{y}$. 
Then, each sub-range is uniformly discretized to a set containing $m$ scores. 
The exponential mechanism is used to select a score $y'_i$ from the $m$ scores in each sub-range $i$ to replace $y_i$, where $1\leq i\leq k$. 
In this way, the order of the scores in $\mathbf{y}$ can be preserved in $\mathbf{y}'$ after replacement. 
As the user of the target model will select the predicted class with the highest score, our method can guarantee zero accuracy loss. 
%In addition, as this replacement is conducted based on an exponential mechanism, the privacy of the $k$ scores in $\mathbf{y}$ can also be preserved. 
In addition, the aim of discretizing each sub-range is to enable the use of an exponential mechanism to select a value. Since the exponential mechanism guarantees differential privacy, the attacker cannot deduce the real value from the selected value. By comparison, other selection methods, e.g., uniformly selection, may not guarantee differential privacy.

In the second phase, since the $k$ scores in $\mathbf{y}'$ may not sum to $1$, 
we use the exponential mechanism again to normalize them to form a valid confidence score vector $\mathbf{z}$. 
There are two reasons to use an exponential mechanism for normalization. First, with exponential normalization, the utility of vectors can be adjusted by choosing different $\epsilon$ values, which allows us to precisely control the difference in the content of information between a confidence score vector and its normalized version. This control, however, may not be achieved using other normalization methods, e.g., the softmax. The second reason is that by using the exponential normalization, only one parameter $\epsilon$ needs to be tuned to achieve both the differentially private modifications and the normalization, which matches the title of the paper: one parameter defense.

%which allows us to precisely control the difference between the information amount contained in the original confidence score vector and the information amount contained in the normalized vector.
We use an example to explain how the method works and guarantees zero classification accuracy loss of the target model. 
Suppose we have a confidence score vector with two scores: $\mathbf{y}=\langle y_1=0.2,y_2=0.8\rangle$. 
First, we divide the range $[0,1)$ into $k=2$ non-overlapping sub-ranges 
based on the scores in $\mathbf{y}$ as $[0,\frac{0.2+0.8}{2})$ and $[\frac{0.2+0.8}{2},1)$.
%First, we equally divide the range $[0,1)$ into $k=2$ non-overlapping sub-ranges, $[0,0.5)$ and $[0.5,1)$. 
Then, sub-range $[0,0.5)$ is matched with $y_1$, while sub-range $[0.5,1)$ is matched with $y_2$, recalling that a sub-range with smaller values is matched with a smaller score. After that, each sub-range is discretized to a set containing $m$ scores. Suppose $m$ is set to $5$, then, sub-range $[0,0.5)$ becomes $\{0,0.1,0.2,0.3,0.4\}$ and sub-range $[0.5,1)$ becomes $\{0.5,0.6,0.7,0.8,0.9\}$. For each score in $\mathbf{y}$, we use the exponential mechanism to select a score in the corresponding discretized sub-range to replace that score. Thus, for $y_1$, we use the exponential mechanism to select a score in $\{0,0.1,0.2,0.3,0.4\}$ to replace the score of $y_1$. Similarly, for $y_2$, the replacing score is selected in $\{0.5,0.6,0.7,0.8,0.9\}$. The detail of the selection will be given in Sub-section \ref{sub:modify}. Let the scores selected in $\{0,0.1,0.2,0.3,0.4\}$ and $\{0.5,0.6,0.7,0.8,0.9\}$ be $y_1'$ and $y_2'$, respectively, thus we have $y_1'<y_2'$. Hence, when we use $y_1'$ and $y_2'$ to replace $y_1$ and $y_2$, respectively, the order of the scores in $\mathbf{y}$ can be preserved in $\mathbf{y}'$ after replacement. This means that the highest score in $\mathbf{y}$, after replacement, is still the highest in $\mathbf{y}'$. Next, suppose the selected scores are $y_1'=0.4$ and $y_2'=0.9$. Note that $0.4+0.9\neq 1$. Then, we use the exponential mechanism again to normalize $\mathbf{y}'$. The aim of the normalization is to ensure that $\sum^k_{i=1} z_i=1$, where each $z_i$ is computed based on the scores in $\mathbf{y}'$. The detail of the normalization will be given in Sub-section \ref{sub:normalize}. This normalization is a requirement of machine learning classifiers, i.e., presenting normalized vectors to users. The normalization result is, say, $\mathbf{z}=\langle 0.3,0.7\rangle$, which is given to the attacker. In Section \ref{sec:analysis}, we will prove that the normalization can preserve the order of scores in $\mathbf{y}'$ in the normalized vector $\mathbf{z}$. As the predicted class is selected only with the highest score, our method can guarantee zero accuracy loss. 
%which cover $0.2$ and $0.8$, respectively. 
%Then, we use the exponential mechanism to modify the two scores $\langle 0.2,0.8\rangle$. 
%The result is, say, $\langle 0.4,0.9\rangle$. 
%Note that $0.4$ is within $[0,0.5)$, $0.9$ is within $[0.5,1)$, but also $0.4+0.9\neq 1$. 
%Second, we use the exponential mechanism again to normalize $\langle 0.4,0.9\rangle$. 
%The result is, say, $\langle 0.3,0.7\rangle$, which is given to the end-user. 

\subsection{Phase 1: Modify the confidence score vector}\label{sub:modify}
\begin{algorithm}
\caption{Division and modification}
\label{alg:step1}
\textbf{Input}: A confidence score vector $\mathbf{y}=\langle y_1,...,y_k\rangle$;\\
\textbf{Output}: A modified confidence score vector $\mathbf{y}'=\langle y'_1,...,y'_k\rangle$;\\
Sort $y_1,...,y_k$ such that $y_1\leq...\leq y_k$;\\
Divide range $[0,1)$ into $k$ sub-ranges based on $y_1,...,y_k$: 
$[0,\frac{y_1+y_2}{2})$, $[\frac{y_1+y_2}{2},\frac{y_2+y_3}{2})$, ..., $[\frac{y_{k-1}+y_k}{2},1)$;\\
%$[0,\frac{1}{k})$, $[\frac{1}{k},\frac{2}{k})$, ..., $[\frac{k-1}{k},1)$;}\\
\For{$i=1$ to $k$}{
%Uniformly discretize sub-range $[\frac{i-1}{k},\frac{i}{k})$ to 
%$\{\frac{i-1}{k},\frac{i-1}{k}+\rho,...,\frac{i-1}{k}+(m-1)\rho\}$, 
%where $\rho=\frac{1}{k\cdot m}$ and $m$ is a positive integer 
Uniformly discretize sub-range $[\frac{y_{i-1}+y_i}{2},\frac{y_i+y_{i+1}}{2})$ 
to $\{\frac{y_{i-1}+y_i}{2},\frac{y_{i-1}+y_i}{2}+\rho,...,\frac{y_{i-1}+y_i}{2}+(m-1)\rho\}$,
where $\rho=\frac{y_{i+1}-y_{i-1}}{2m}$ and $m$ is a positive integer representing the granularity of discretization;} 
\For{each $y_i$ from $y_1$ to $y_k$}{
Use the exponential mechanism to randomly select a value 
from the corresponding dicretized sub-range $i$ as the modified score, $y'_i$;\\
}
\end{algorithm}

The first phase of our method is formalized in Algorithm \ref{alg:step1}. 
In summary, the method takes a confidence score vector $\mathbf{y}$ as input 
and outputs a modified vector $\mathbf{y}'$. 
In Line 4, the range $[0,1)$ is divided into $k$ non-overlapping sub-ranges based on $y_1,...,y_k$.  
%each sub-range covering a score in vector $\mathbf{y}$. 
In Lines 5 and 6, each sub-range is uniformly discretized to a set containing $m$ scores. 
In Lines 7 and 8, the exponential mechanism is used to randomly select a score $y'_i$ in each sub-range $i$ to replace $y_i$. 
As described in Definition \ref{Def-Ex}, which score is selected in each sub-range $i$ is based on its utility. 
The utility of the $j^{th}$ score is then set to 
\begin{equation}\nonumber
\begin{aligned}
    u^i_j&=\frac{1}{|y_i-(\frac{i-1}{k}+(j-1)\rho)|}.\\
    %&=\frac{1}{|\frac{y_i-y_{i-1}}{2}-(j-1)\rho|}.
\end{aligned}
\end{equation}
%\begin{equation}\nonumber
%\begin{aligned}
%    u^i_j&=\frac{1}{|y_i-(\frac{y_{i-1}+y_i}{2}+(j-1)\rho)|}\\
%    &=\frac{1}{|\frac{y_i-y_{i-1}}{2}-(j-1)\rho|}.
%\end{aligned}
%\end{equation}
Scores with a smaller difference to $y_i$ have a higher utility and thus a higher probability of being selected. 
Thus, in the sub-range $i$, the probability of selecting the $j^{th}$ score is proportional to $exp(\frac{\epsilon u^i_j}{2\Delta u})$. 
%and $\Delta u=\frac{y_i+y_{i+1}}{2}-\frac{y_{i-1}+y_i}{2}=\frac{y_{i+1}-y_{i-1}}{2}$ is the sensitivity of the utility score.

In Algorithm \ref{alg:step1}, we need to conduct $k$ samplings, i.e., selections, and the size of each sampling domain is $m$. To avoid a large computation overhead, a finite and small sampling domain is necessary \cite{Ganesh20}. To limit the size of the sampling domain, we can tune the value of $m$. In the experiments, we set $m=5$, i.e., the final result is selected from $5$ candidates using the exponential mechanism. The set of $m=5$ yields only a small sampling domain and does not introduce a large computation overhead.

\subsection{Phase 2: Normalizing the confidence score vector}\label{sub:normalize}
The modifications to the confidence score vector $\mathbf{y}$ in Phase 1 may result in an invalid probability distribution, 
i.e., where the sum of the scores in $\mathbf{y}'$ does not equal $1$. 
Hence, Phase 2 involves using the exponential mechanism again to normalize $\mathbf{y}'$. 
This procedure is as follows. 

Let the modified vector be $\mathbf{y}'=\langle y'_1,...,y'_k\rangle$. 
%As $\mathbf{y}$ is a confidence score vector, 
%we have: $\sum_{1\leq i\leq k}y_i=1$ and each $y_i\in (0,1)$. 
Then, $\mathbf{y}'$ can be perturbed into $\mathbf{z}$, 
where each $z_i$ is computed using the following equation.
\begin{equation}\label{eq:perturb}
    z_i=\frac{exp(\frac{\epsilon u(\mathbf{y}',y'_i)}{2\Delta u})}{\sum_{1\leq j\leq k}exp(\frac{\epsilon u(\mathbf{y}',y'_j)}{2\Delta u})}
\end{equation}
Theoretically, $u(\mathbf{y}',y'_i)$ can be set to any value which is positively correlated with $y'_i$, i.e., a larger $y'_i$ should be assigned a larger $u(\mathbf{y}',y'_i)$. Specifically, by setting $u(\mathbf{y}',y'_i)=y'_i$, i.e., setting the utility of each confidence score to the same value as the score, the sensitivity $\Delta u$ becomes $1$ which is easy for both theoretically analyzing the properties of our defense method and experimentally evaluating its performance.
%If we set $u(\mathbf{y}',y'_i)=y'_i$, i.e., if we set the probability to the same value as the utility of each confidence score, then we have $\Delta u=1$. 
Therefore, Equation \ref{eq:perturb} can be simplified to 
\begin{equation}\label{eq:perturb2}
    z_i=\frac{exp(\frac{\epsilon y'_i}{2})}{\sum_{1\leq j\leq k}exp(\frac{\epsilon y'_j}{2})}.
\end{equation}
The resulting vector $\mathbf{z}$ is the final output. %given to the end-user as

\subsection{Discussion of the method}
In our method, the privacy protection is on confidence score vectors $\mathbf{y}$ rather than the original training dataset $D^{\rm train}_{\rm target}$. This is because 1) attackers in our problem are not allowed to directly access the training dataset, and 2) model inversion attacks are not against the training dataset, instead, these attacks aim to reconstruct any input data to the target model. The attackers can access the target model and exploit the prediction results of the target model to launch both membership inference and model inversion attacks. Therefore, our protection focuses on the prediction results of the target model, i.e., confidence score vectors.

Although our method is applied only to the prediction results of the target model, it can still defend against both membership inference and model inversion attacks. The attacker queries the target model $F$ by feeding a data record $\mathbf{x}$ to it, and expects to receive a response $F(\mathbf{x})$. However, by using our private prediction interface $\mathcal{M}_d$, the response $F(\mathbf{x})$ is obfuscated to $\mathcal{M}_d(F(\mathbf{x}))$. Since the existing attack methods must use $F(\mathbf{x})$ to launch the membership inference and model inversion attacks, altering $F(\mathbf{x})$ to $\mathcal{M}_d(F(\mathbf{x}))$ can significantly reduce the attack precision. This can be further explained by the fact that the attack model $G$ is trained using $F(\mathbf{x})$ as input, i.e., $G(F(\mathbf{x}))$, where for membership inference attacks, $G(F(\mathbf{x}))=Pr(\mathbf{x}\in D^{\rm train}_{\rm target})$, and for model inversion attacks, $G(F(\mathbf{x}))=\hat{x}$. Now, $F(\mathbf{x})$ is obfuscated to $\mathcal{M}_d(F(\mathbf{x}))$ which is used as input to the attack model $G$, i.e., $G(\mathcal{M}_d(F(\mathbf{x})))$. To achieve a precise attack, the attacker has to guarantee that $G(\mathcal{M}_d(F(\mathbf{x})))=G(F(\mathbf{x}))$, which means that $\mathcal{M}_d(F(\mathbf{x}))$ and $F(\mathbf{x})$ must be in the same class. However, there is no guarantee that an input can be in the same class as its obfuscated version with differentially private noise. Moreover, as $\mathcal{M}_d$ satisfies differential privacy, the attacker cannot deduce $F(\mathbf{x})$ from $\mathcal{M}_d(F(\mathbf{x}))$. In the experiments, we have also attempted to train the attack model $G$ using $\mathcal{M}_d(F(\mathbf{x}))$ as input, but this training did not converge.

\section{Properties of the defense method and how to use them to defend against attacks}\label{sec:analysis}
This section begins with the proof that our defense method satisfies differential privacy. 
Also, to maintain the differential privacy guarantee, a bound is set to limit the number of queries 
that an attacker can access the target model $F$ using the same input data record $\mathbf{x}$ (Sub-section \ref{sub:privacy}). 
Then, an analysis of tuning $\epsilon$ follows, which details the various properties of our method with different $\epsilon$ values (Sub-section \ref{sub:epsilon}). 
After that, we explain how to use these properties to defend against the data inference attacks (Sub-sections \ref{sub:defend inference} and \ref{sub:defend inversion}).

\subsection{Privacy analysis}\label{sub:privacy}
\begin{lem}[Post-processing theorem \cite{Dwork14}]\label{lem:post-processing}
Let $\mathcal{M}_1:Y\rightarrow Y'$ be a randomized algorithm 
that is $\epsilon$-differentially private. 
Let $\mathcal{M}_2:Y'\rightarrow Z$ be any mapping including deterministic functions. 
Then $\mathcal{M}_1\circ\mathcal{M}_2:Y\rightarrow Z$ is $\epsilon$-differentially private.
\end{lem}
Lemma \ref{lem:post-processing} states that the combination of a differentially private algorithm 
and a deterministic algorithm still guarantees differential privacy.

\begin{thm}\label{thm:DP}
%The defense method, $\mathcal{M}_d$, 
Algorithm \ref{alg:step1} satisfies $(k\cdot\epsilon)$-differential privacy, where $\epsilon$ is the privacy budget and $k$ is 
the number of scores in a confidence score vector $\mathbf{y}$, i.e., the number of classes that the target model $F$ can classify.
\end{thm}
\begin{proof}
In Line 8 of Algorithm \ref{alg:step1} (the first phase of our method), the exponential mechanism is used $k$ times 
to randomly select $k$ scores to replace $y_1,...,y_k$, i.e., to perform the mapping from $\mathbf{y}$ to $\mathbf{y}'$. 
According to Definition \ref{Def-Ex}, exponential mechanism defines a utility function $u:\mathbb{N}^{|\mathcal{X}|}\times\mathcal{R}\rightarrow\mathbb{R}$, which maps dataset-output pairs to utility scores. 
In our problem, for each sub-range $i$: $[\frac{y_{i-1}+y_i}{2},\frac{y_i+y_{i+1}}{2})$, when $2\leq i\leq k-1$, 
and $[0,\frac{y_1+y_2}{2})$ and $[\frac{y_{k-1}+y_k}{2},1)$ when $i=1$ and $i=k$, respectively, 
%$[\frac{i-1}{k},\frac{i}{k})$, where $1\leq i\leq k$, 
by applying Definition \ref{Def-Ex}, we interpret dataset $D$ as a confidence score vector $\mathbf{y}$, 
and interpret output set $\mathcal{R}$ as $\mathcal{R}=\{\frac{y_{i-1}+y_i}{2},\frac{y_{i-1}+y_i}{2}+\rho,...,\frac{y_{i-1}+y_i}{2}+(m-1)\rho\}$. 
%$\mathcal{R}=\{\frac{i-1}{k},\frac{i-1}{k}+\rho,...,\frac{i-1}{k}+(m-1)\rho\}$. 
%Here, the output set is identical to the dataset, because the output value is selected from the dataset. 
In addition, we interpret a neighboring dataset $D'$ as another confidence score vector $\hat{\mathbf{y}}$, 
where $||\hat{\mathbf{y}}-\mathbf{y}||_1\leq 1$ and $y_i+y_{i+1}=\hat{y}_i+\hat{y}_{i+1}$, where $1\leq i\leq k-1$. 
Then, $\mathbf{y}$ and $\hat{\mathbf{y}}$ have the same output set $\mathcal{R}$. 

%which means that one extra element is added to $D$ or one existing element is removed from $D$. 
Let the probability of selecting the $j^{th}$ score $r_j$ from $\mathcal{R}$ with vector $\mathbf{y}$ be $P(r_j|\mathbf{y})$ 
and the probability of selecting $r_j$ from $\mathcal{R}$ with a neighboring vector $\hat{\mathbf{y}}$ be $P(r_j|\hat{\mathbf{y}})$. 
Also, let $u^i(\mathbf{y},r_j)$ and $u^i(\hat{\mathbf{y}},r_j)$ be the utility of of selecting $r_j$ 
from $\mathcal{R}$ with $\mathbf{y}$ and $\hat{\mathbf{y}}$, respectively. 
Then, we have $\frac{P(r_j|\mathbf{y})}{P(r_j|\hat{\mathbf{y}})}=\Big[\frac{exp(\frac{\epsilon u^i(\mathbf{y},r_j)}{2\Delta u})}{\sum_{r_{j'}\in\mathcal{R}}exp(\frac{\epsilon u^i(\mathbf{y},r_{j'})}{2\Delta u})}\Big]\Big/\Big[\frac{exp(\frac{\epsilon u^i(\hat{\mathbf{y}},r_j)}{2\Delta u})}{\sum_{r_{j'}\in\mathcal{R}}exp(\frac{\epsilon u^i(\hat{\mathbf{y}},r_{j'})}{2\Delta u})}\Big]$. Based on the knowledge of exponential mechanism \cite{Dwork14}, we have $\frac{P(r_j|\mathbf{y})}{P(r_j|\hat{\mathbf{y}})}\leq exp(\epsilon)$.

By symmetry, $\frac{P(r_j|\mathbf{y})}{P(r_j|\hat{\mathbf{y}})}\geq exp(-\epsilon)$. 
Therefore, in each sub-range $i$, our method satisfies $\epsilon$-differential privacy. 
Since there are $k$ sub-ranges and the exponential mechanism is used in each sub-range, 
Algorithm \ref{alg:step1} gives $(k\cdot\epsilon)$-differential privacy. 

The second phase of our method is deterministic, mapping a differentially private vector $\mathbf{y}'$ to $\mathbf{z}$. 
Therefore, according to Lemma \ref{lem:post-processing}, 
the combination of the first and second phases still guarantee $(k\cdot\epsilon)$-differential privacy. 
\end{proof}

\begin{thm}\label{thm:PPI}
The defense method, $\mathcal{M}_d$, is a $(k\cdot\epsilon)$-differentially private prediction interface.
\end{thm}
\begin{proof}
Theorem \ref{thm:DP} shows that for each input query $\mathbf{x}$ of the target model $F$, 
the defense method $\mathcal{M}_d$ can obfuscate the corresponding response, i.e., output vector $\mathbf{y}$, in a differentially private manner. 
According to Definition \ref{def:PPI}, as the sequence of queries and responses satisfies $(k\cdot\epsilon)$-differential privacy, 
the defense method $\mathcal{M}_d$ is a $(k\cdot\epsilon)$-differentially private prediction interface.
\end{proof}

Theorems \ref{thm:DP} and \ref{thm:PPI} and the accompanying proof demonstrate that 
our defense method $\mathcal{M}_d$ provides a privacy guarantee. 
By giving differential privacy guarantee, an attacker cannot deduce the original vector $\mathbf{y}$ from the perturbed vector $\mathbf{z}$. 
As shown in the experiments, based on $\mathbf{z}$, the attacker can neither precisely reconstruct the input data record $\mathbf{x}$ 
nor successfully infer whether $\mathbf{x}$ is in the training set of the target model $F$. 
However, if an attacker uses the same input data record $\mathbf{x}$ to access the target model $F$ multiple times, 
the attacker may deduce the original vector $\mathbf{y}$ by observing the perturbed vectors. 
This is because in differential privacy, a privacy budget is used to control the privacy level. 
Every time an original vector is perturbed and released, the privacy budget is partially consumed. 
Once the privacy budget is used up, differential privacy cannot guarantee the privacy of the original vector anymore.
To guarantee the privacy level of an original vector, a bound must be set on the number of times 
that a user can access the target model using the same data record. 
The detailed computation of the bound is as follows. 

%One may concern that a large number of attempts by an attacker may render our method failing to maintain the differential privacy guarantee. 
%This is an issue in general differential privacy if these attempts are made to the same dataset. 
%However, in our problem, an attacker is dealing with confidence score vectors, . 
%As every time the confidence score vector is different, the dataset is also different. 
%Hence, a large number of queries does not affect the differential privacy guarantee of our method. 

\begin{definition}[KL-Divergence \cite{Dwork14}]\label{Def:KL}
The KL-Divergence between two random variables $Y$ and $Z$ taking values from the same domain is defined to be:
\begin{equation}
D(Y||Z)=\mathbb{E}_{y\sim Y}\left[ln\frac{Pr(Y=y)}{Pr(Z=y)}\right].
\end{equation}
\end{definition}

\begin{definition}[Max Divergence \cite{Dwork14}]\label{Def:Max}
The Max Divergence between two random variables $Y$ and $Z$ taking values from the same domain is defined to be:
\begin{equation}
D_{\infty}(Y||Z)=\max\limits_{S\subseteq Supp(Y)}\left[ln\frac{Pr(Y\in S)}{Pr(Z\in S)}\right].
\end{equation}
\end{definition}

\begin{lem}[\cite{Dwork14}]\label{lem:DPDivergence}
A mechanism $\mathcal{M}$ is $\epsilon$-differentially private if and only if on every two neighboring datasets $x$ and $x'$,
$D_{\infty}(\mathcal{M}(x)||\mathcal{M}(x'))\leq\epsilon$ and $D_{\infty}(\mathcal{M}(x')||\mathcal{M}(x))\leq\epsilon$.
\end{lem}

\begin{lem}[\cite{Dwork14}]\label{lem:KLMax}
Suppose that random variables $Y$ and $Z$ satisfy $D_{\infty}(Y||Z)\leq\epsilon$ and $D_{\infty}(Z||Y)\leq\epsilon$. 
Then, $D(Y||Z)\leq\epsilon\cdot (e^{\epsilon}-1)$.
\end{lem}

\begin{thm}\label{thm:bound}
Given that the privacy level of an original vector is $k\cdot\epsilon$ for each access, 
to guarantee its overall privacy level to be $\epsilon'$,
the upper bound of the number of rounds is $\frac{\epsilon'\cdot(e^{\epsilon'}-1)}{k\cdot\epsilon\cdot(e^{k\epsilon}-1)}$.
\end{thm}
\begin{proof}
Let the upper bound of the number of access times be $b$, and the corresponding $b$ perturbed vectors, which can be observed by an attacker, be $\mathbf{o}=(\mathbf{z}^1,...,\mathbf{z}^b)$. 
We have 
\begin{equation}\nonumber
\begin{aligned}
    D(Y||Z)&=ln\left[\frac{Pr(Y=\mathbf{o})}{Pr(Z=\mathbf{o})}\right]=ln\left[\prod^b_{i=1}\frac{Pr(Y_i=\mathbf{z}^i)}{Pr(Z_i=\mathbf{z}^i)}\right]\\
    &=\sum^b_{i=1}ln\left[\frac{Pr(Y_i=\mathbf{z}^i)}{Pr(Z_i=\mathbf{z}^i)}\right]=\sum^b_{i=1}D(Y_i||Z_i).
\end{aligned}
\end{equation}
As the original vector is guaranteed $k\cdot\epsilon$-differential privacy for each individual access, 
based on Lemma \ref{lem:DPDivergence}, 
we have $D_{\infty}(Y_i||Z_i)\leq k\cdot\epsilon$ and $D_{\infty}(Z_i||Y_i)\leq k\cdot\epsilon$.
Based on this result, according to Lemma \ref{lem:KLMax},
we have $D(Y_i||Z_i)\leq k\cdot\epsilon\cdot(e^{k\epsilon}-1)$.
Thus, we have $D(Y||Z)=\sum^b_{i=1}D(Y_i||Z_i)\leq b\cdot k\cdot\epsilon\cdot(e^{k\epsilon}-1)$.

To guarantee the original vector's overall privacy level to be $\epsilon'$,
according to Lemma \ref{lem:DPDivergence},
we have $D_{\infty}(Y||Z)\leq\epsilon'$ and $D_{\infty}(Z||Y)\leq\epsilon'$.
By using Lemma \ref{lem:KLMax},
we have $D(Y||Z)\leq\epsilon'\cdot(e^{\epsilon'}-1)$.
As $D(Y||Z)\leq b\cdot k\cdot\epsilon\cdot(e^{k\epsilon}-1)$ and $D(Y||Z)\leq\epsilon'\cdot(e^{\epsilon'}-1)$,
the upper bound $b$ is limited by $\frac{\epsilon'\cdot(e^{\epsilon'}-1)}{k\cdot\epsilon\cdot(e^{k\epsilon}-1)}$.
\end{proof}

\subsection{Analysis of tuning $\epsilon$}\label{sub:epsilon}
The input for the first step of our method is a confidence score vector $\mathbf{y}$, 
and the output is a perturbed vector $\mathbf{y}'$. 
The input for the second step is $\mathbf{y}'$, 
and the output is a normalized vector $\mathbf{z}$. 
To maximize the utility of the normalized vector, 
it should be equal to the original vector: $\mathbf{z}=\mathbf{y}$. 
This equality does not guarantee any privacy over the original vector and must be avoided in practice, 
but this equality gives us a start point to investigate how to tune the parameter, $\epsilon$, to achieve the balance between privacy and utility. 
Specifically, the utility of a normalized vector $\mathbf{z}$ is defined as $||\mathbf{z}-\mathbf{y}||_1$, 
where a smaller value of $||\mathbf{z}-\mathbf{y}||_1$ means a higher utility but a lower privacy guarantee. 
%Hence, our analysis begins with this equality in an attempt to balance privacy and utility. 

Given that $\mathbf{y}=\langle y_1,...,y_k\rangle$, $\mathbf{y}'=\langle y'_1,...,y'_k\rangle$ 
and $\mathbf{z}=\langle z_1,...,z_k\rangle$, we have the following system of $k$ equations:
\begin{equation}\label{eq:system}
\left\{
\begin{array}{lr}
    z_1=\frac{exp(\frac{\epsilon y'_1}{2})}{\sum_{1\leq j\leq k}exp(\frac{\epsilon y'_j}{2})}=y_1, &  \\
    ... \\
    z_k=\frac{exp(\frac{\epsilon y'_k}{2})}{\sum_{1\leq j\leq k}exp(\frac{\epsilon y'_j}{2})}=y_k. &  
\end{array}
\right.
\end{equation}
We first prove that Algorithm \ref{alg:step1} preserves the order of scores in vector $\mathbf{y}$, 
and then show that Equation \ref{eq:system} has a unique positive solution. 
After that, we analyze the properties of the solution to Equation \ref{eq:system}.
%The solution to the system is given next, followed by an analysis of its properties.

%\subsubsection{The solution to Equation \ref{eq:system}}
\begin{lem}\label{lem0}
Algorithm \ref{alg:step1} preserves the order of scores in vector $\mathbf{y}$, 
i.e., if $y_i<y_j$, then $y'_i<y'_j$, where $1\leq i\neq j \leq k$.
\end{lem}
\begin{proof}
According to Line 3 of Algorithm \ref{alg:step1}, 
when $y_i<y_j$, we have $i<j$.
According to Line 8 of Algorithm \ref{alg:step1}, we have 
%$y'_i\in[\frac{i-1}{k},\frac{i}{k})$ and 
%$y'_j\in[\frac{j-1}{k},\frac{j}{k})$.
$y'_i\in[\frac{y_{i-1}+y_i}{2},\frac{y_i+y_{i+1}}{2})$ and 
$y'_j\in[\frac{y_{j-1}+y_j}{2},\frac{y_j+y_{j+1}}{2})$. 
As $i<j$, we have $i\leq j-1$ and $i+1\leq j$. 
Thus, we have $y_i\leq y_{j-1}$ and $y_{i+1}\leq y_j$, which implies that $\frac{y_i+y_{i+1}}{2}\leq \frac{y_{j-1}+y_j}{2}$. 
Therefore, we can conclude $y'_i<y'_j$.
\end{proof}

The following Lemmas \ref{lem1} and \ref{lem2} and Theorem \ref{thm:unique} show that 
Equation \ref{eq:system} has a unique positive solution.
\begin{lem}\label{lem1}
%We assume that at least one score in $\mathbf{y}$ is not equal to $\frac{1}{k}$. 
Let $y'_{min}=min\{y'_1,...,y'_k\}$ and $y'_{max}=max\{y'_1,...,y'_k\}$, 
we have: $y'_{min}<\frac{1}{k}$ and $y'_{max}\geq\frac{1}{k}$.
\end{lem}
\begin{proof}
The following proof covers $y'_{max}\geq\frac{1}{k}$. 
The proof of $y'_{min}<\frac{1}{k}$ is similar. 

According to Lines 4-6 of Algorithm \ref{alg:step1}, we know that $\frac{y_{k-1}+y_k}{2}\leq y'_{max}<1$.
%$\frac{k-1}{2}\leq y'_{max}<1$. 
%Thus, to prove $y'_{max}\geq\frac{1}{k}$, we need only to prove $\frac{k-1}{k}\geq\frac{1}{k}$. 
Thus, to prove $y'_{max}\geq\frac{1}{k}$, we need only to prove $\frac{y_{k-1}+y_k}{2}\geq\frac{1}{k}$. 
For this, we use mathematical induction. 
When $k=2$, we have $\frac{y_{k-1}+y_k}{2}=\frac{y_1+y_2}{2}=\frac{1}{2}=\frac{1}{k}$. 
Thus, the conclusion is established. 
Assume that when $k=n$, the conclusion is also established, 
i.e., $\frac{y_{n-1}+y_n}{2}\geq\frac{1}{n}$. 
Next, we prove that when $k=n+1$, the conclusion is still established, 
i.e., proving $\frac{y_n+y_{n+1}}{2}\geq\frac{1}{n+1}$. 
As $\frac{y_{n-1}+y_n}{2}\geq\frac{1}{n}$, we have $y_n\geq\frac{2}{n}-y_{n-1}$. 
Therefore, we have
\begin{equation}\nonumber
\begin{aligned}
    \frac{y_n+y_{n+1}}{2}&\geq\frac{(\frac{2}{n}-y_{n-1})+y_{n+1}}{2} 
    \\&=\frac{1}{n}+\frac{y_{n+1}-y_{n-1}}{2}
    \\&\geq\frac{1}{n}>\frac{1}{n+1}.
\end{aligned}
\end{equation}
The second inequality is based on Line 3 of Algorithm \ref{alg:step1}, 
where the scores of vector $\mathbf{y}$ are sorted as $y_1\leq...\leq y_k$. 
Hence, $y_{n+1}\geq y_{n-1}$ and $\frac{y_{n+1}-y_{n-1}}{2}\geq 0$. 
The lemma has been proven.
    
%\textcolor{blue}{If $y_{min}=y_{max}=\frac{1}{k}$, all the scores in $\mathbf{y}$ equal to $\frac{1}{k}$, 
%which contradicts the assumption. 
%Thus, $y_{min}<y_{max}$. 
%Moreover, if $y_{min}\geq\frac{1}{k}$, then $\sum_{1\leq j\leq k}y_k>1$, 
%which contradicts that $\mathbf{y}$ is a confidence score vector: $\sum_{1\leq j\leq k}y_k=1$. 
%Thus, $y_{min}<\frac{1}{k}$. 
%By symmetry, $y_{max}>\frac{1}{k}$.}
%For this, we have $\frac{k-1}{k}-\frac{1}{k}=\frac{k-2}{k}$. 
%In our problem, the target model $F$ can classify $k\geq 2$ classes. 
%Therefore, $k-2\geq 0$ and thus, $\frac{k-2}{k}\geq 0\Rightarrow\frac{k-1}{k}\geq\frac{1}{k}$. 
%The equality holds if and only if $k=2$.
\end{proof}

\begin{lem}\label{lem2}
The system in Equation \ref{eq:system} has at least one positive solution.
\end{lem}
\begin{proof}
According to Equation \ref{eq:system}, we have 
$\frac{y_{max}}{y_j}=\frac{exp(\frac{\epsilon y'_{max}}{2})}{exp(\frac{\epsilon y'_j}{2})}$, 
where $1\leq j\leq k$. Then, we have the following deduction. 
\begin{equation}\nonumber
\begin{aligned}
    \frac{y_{max}}{y_j}&=\frac{exp(\frac{\epsilon y'_{max}}{2})}{exp(\frac{\epsilon y'_j}{2})} 
    \\ln(\frac{y_{max}}{y_j})&=ln(\frac{exp(\frac{\epsilon y'_{max}}{2})}{exp(\frac{\epsilon y'_j}{2})})
    \\ln(y_{max})-ln(y_j)&=\frac{\epsilon}{2}(y'_{max}-y'_j) 
    %\\ \sum_{1\leq j\leq k}[ln(y_{max})-ln(y_j)]&=\sum_{1\leq j\leq k}\frac{\epsilon}{2}(y'_{max}-y'_j)
    %\\ k\cdot ln(y_{max})-\sum_{1\leq j\leq k}ln(y_j)&=\frac{\epsilon}{2}(k\cdot y'_{max}-\sum_{1\leq j\leq k}y'_j)
    %\\ k\cdot ln(y_{max})-ln(\prod_{1\leq j\leq k}y_j)&=\frac{\epsilon}{2}(k\cdot y'_{max}-1)
    \\ \epsilon&=\frac{2[ln(y_{max})-ln(y_j)]}{y'_{max}-y'_j}>0
\end{aligned}
\end{equation}
%According to Lemma \ref{lem1}, $y'_{max}>\frac{1}{k}$, thus $k\cdot y'_{max}>1$.
%Then, we have: 
%\begin{equation}\nonumber
%\begin{aligned}
%    \epsilon=\frac{2[k\cdot ln(y_{max})-ln(\prod_{1\leq j\leq k}y_j)]}{k\cdot y_{max}-1}>0.
%\end{aligned}
%\end{equation}
This $\epsilon$ is a positive solution to the system.
\end{proof}

Next, we prove that this $\epsilon$ is a unique solution. 
Formally, based on Lemmas \ref{lem1} and \ref{lem2}, we have: 
\begin{thm}\label{thm:unique}
The system in Equation \ref{eq:system} has a unique positive solution.
\end{thm}
\begin{proof}
Assume we have two different positive solutions $\epsilon_1$ and $\epsilon_2$. 
Then, we have 
\begin{equation}\nonumber
\begin{aligned}
    &\frac{y_{max}}{y_{min}}=\frac{exp(\frac{\epsilon_1 y'_{max}}{2})}{exp(\frac{\epsilon_1 y'_{min}}{2})}=\frac{exp(\frac{\epsilon_2 y'_{max}}{2})}{exp(\frac{\epsilon_2 y'_{min}}{2})}
    \\& \Rightarrow exp[\frac{\epsilon_1-\epsilon_2}{2}y'_{max}]=exp[\frac{\epsilon_1-\epsilon_2}{2}y'_{min}]
    \\& \Rightarrow exp[\frac{\epsilon_1-\epsilon_2}{2}(y'_{max}-y'_{min})]=1.
\end{aligned}
\end{equation}
According to Lemma \ref{lem1}, $y'_{max}\neq y'_{min}$, thus $\epsilon_1=\epsilon_2$.
According to Lemma \ref{lem2}, the system has at least one solution. 
Hence, the system has a unique solution.
\end{proof}

The properties of the solution are analyzed as follows. 
Given $z_i=y_i=\frac{exp(\frac{\epsilon y'_i}{2})}{\sum_{1\leq j\leq k}exp(\frac{\epsilon y'_j}{2})}$, 
where $1\leq i\leq k$, the solution has the following properties. 
%\begin{property}\label{property1}
%When $\epsilon=0$, then $z_i=\frac{1}{k}$ for all $1\leq i\leq k$.
%\end{property}
%\begin{proof}
%This property can be observed by directly computing $z_i=\frac{exp(\frac{\epsilon y'_i}{2})}{\sum_{1\leq j\leq k}exp(\frac{\epsilon y'_j}{2})}$.
%\end{proof}
%Property \ref{property1} states that $\epsilon=0$ gives the strongest privacy guarantee 
%of the confidence score vector but at the cost of all its utility.

\begin{property}\label{property2}
When $\epsilon>0$, if $y'_i>y'_j$, then $z_i>z_j$, where $1\leq i, j\leq k$.
\end{property}
\begin{proof}
$\frac{z_i}{z_j}=exp[\frac{\epsilon}{2}(y'_i-y'_j)]\geq 1$,  
therefore, $z_i>z_j$.
\end{proof}
Property \ref{property2} combined with Lemma \ref{lem0} contends that 
our defense method preserves the order of scores in confidence score vector $\mathbf{y}$. 
As such, it holds that since the order of scores in $\mathbf{y}$ can be preserved, 
when $\epsilon>0$, 
the utility of the confidence score vector $\mathbf{y}$ can be guaranteed. 
This means that the class with the highest probability in $\mathbf{y}$ 
still has the highest probability in the normalized vector $\mathbf{z}$. 
This property guarantees a good user experience, 
as users usually select the predicted class with the highest probability.

Further discussions on Property \ref{property2} are included with the 
Properties \ref{property3} and \ref{property4}.
\begin{property}\label{property3}
Let $\epsilon^*$ be the solution of the system and $y'_i>y'_j$, where $1\leq i,j\leq k$. 
Then, if $\epsilon>\epsilon^*$, $\frac{z_i}{z_j}>\frac{y_i}{y_j}$; 
if $\epsilon<\epsilon^*$, $\frac{z_i}{z_j}<\frac{y_i}{y_j}$.
\end{property}
\begin{proof}
Because $z_i=\frac{exp(\frac{\epsilon y'_i}{2})}{\sum_{1\leq j\leq k}exp(\frac{\epsilon y'_j}{2})}$ and 
$y_i=\frac{exp(\frac{\epsilon^* y'_i}{2})}{\sum_{1\leq j\leq k}exp(\frac{\epsilon^* y'_j}{2})}$, 
we have 
\begin{equation}\nonumber
\begin{aligned}
    \frac{z_i}{z_j}/\frac{y_i}{y_j}&=\frac{exp[\frac{\epsilon}{2}(y'_i-y'_j)]}{exp[\frac{\epsilon^*}{2}(y'_i-y'_j)]}
    \\& =exp[\frac{\epsilon-\epsilon^*}{2}(y'_i-y'_j)].
\end{aligned}
\end{equation}
Since it is assumed that $y'_i>y'_j$, when $\epsilon>\epsilon*$, 
$\frac{z_i}{z_j}/\frac{y_i}{y_j}>1\Rightarrow\frac{z_i}{z_j}>\frac{y_i}{y_j}$; 
when $\epsilon<\epsilon^*$, $\frac{z_i}{z_j}/\frac{y_i}{y_j}<1\Rightarrow\frac{z_i}{z_j}<\frac{y_i}{y_j}$.
\end{proof}

\begin{property}\label{property4}
If $\epsilon>\epsilon^*$, then $z_{min}<y_{min}$ and $z_{max}>y_{max}$; 
if $\epsilon<\epsilon^*$, then $z_{min}>y_{min}$ and $z_{max}<y_{max}$.
\end{property}
\begin{proof}
We know that 
\begin{equation}\nonumber
\begin{aligned}
    &\sum_{1\leq i\leq k}exp[\frac{\epsilon}{2}(y'_i-y'_j)]-\sum_{1\leq i\leq k}exp[\frac{\epsilon^*}{2}(y'_i-y'_j)]
    \\& =\sum_{1\leq i\leq k}[exp[\frac{\epsilon}{2}(y'_i-y'_j)]-exp[\frac{\epsilon^*}{2}(y'_i-y'_j)]],
\end{aligned}
\end{equation}
where $1\leq j\leq k$. 
Therefore, when $\epsilon>\epsilon^*$, we have 
\begin{equation}\nonumber
\begin{aligned}
    &\sum_{1\leq i\leq k}exp[\frac{\epsilon}{2}(y'_i-y'_{min})]-\sum_{1\leq i\leq k}exp[\frac{\epsilon^*}{2}(y'_i-y'_{min})]>0,
    \\& \sum_{1\leq i\leq k}exp[\frac{\epsilon}{2}(y'_i-y'_{max})]-\sum_{1\leq i\leq k}exp[\frac{\epsilon^*}{2}(y'_i-y'_{max})]<0.
\end{aligned}
\end{equation}
Because 
\begin{equation}\nonumber
\begin{aligned}
    &\frac{1}{z_{min}}=\sum_{1\leq i\leq k}\frac{z_i}{z_{min}}=\sum_{1\leq i\leq k}exp[\frac{\epsilon}{2}(y'_i-y'_{min})], 
    \\& \frac{1}{y_{min}}=\sum_{1\leq i\leq k}\frac{y_i}{y_{min}}=\sum_{1\leq i\leq k}exp[\frac{\epsilon^*}{2}(y'_i-y'_{min})], 
    \\& \frac{1}{z_{max}}=\sum_{1\leq i\leq k}\frac{z_i}{z_{max}}=\sum_{1\leq i\leq k}exp[\frac{\epsilon}{2}(y'_i-y'_{max})], 
    \\& \frac{1}{y_{max}}=\sum_{1\leq i\leq k}\frac{y_i}{y_{max}}=\sum_{1\leq i\leq k}exp[\frac{\epsilon^*}{2}(y'_i-y'_{max})],
\end{aligned}
\end{equation}
thus, we have 
\begin{equation}\nonumber
\begin{aligned}
    &\frac{1}{z_{min}}-\frac{1}{y_{min}}>0\Rightarrow z_{min}<y_{min}, 
    \\& \frac{1}{z_{max}}-\frac{1}{y_{max}}<0\Rightarrow z_{max}>y_{max}.
\end{aligned}
\end{equation}
By symmetry, we have that when $\epsilon<\epsilon^*$, 
$z_{min}>y_{min}$ and $z_{max}<y_{max}$.
\end{proof}

Properties \ref{property3} and \ref{property4} state that if $\epsilon>\epsilon^*$, the difference in probabilities $y_{max}-y_{min}$ within the confidence score vector $\mathbf{y}$ will increase in vector $\mathbf{z}$, i.e., $z_{max}-z_{min}>y_{max}-y_{min}$. This means that even if the target model is not very confident in predicting an input data record $\mathbf{x}$, the defense method can make the output appear very confident to the attacker. Similarly, the defense method can also make a very confident output appear less confident to the attacker.

The last property to analyze is the change in the distance between a confidence score vector $\mathbf{y}$ and its perturbed version $\mathbf{z}$. As the success of a model inversion attack is based on the rich information contained in each confidence score vector, an intuitive way to defend against model inversion attacks is to widen the distance between $\mathbf{y}$ and $\mathbf{z}$. The distance is formally defined as $|\mathbf{z}-\mathbf{y}|_1=\sum^k_{i=1}|z_i-y_i|$. Then, we have the following property.

\begin{property}\label{property5}
If $\epsilon>\epsilon^*$, then $|\mathbf{z}-\mathbf{y}|_1$ increases as $\epsilon$ increases; if $\epsilon<\epsilon^*$, then $|\mathbf{z}-\mathbf{y}|_1$ increases as $\epsilon$ decreases.
\end{property}
\begin{proof}

\begin{equation}\nonumber
\begin{aligned}
&|\mathbf{z}-\mathbf{y}|_1=\sum^k_{i=1}|z_i-y_i|\\
&=\sum^k_{i=1}|\frac{exp(\frac{\epsilon y'_i}{2})}{\sum^k_{j=1}exp(\frac{\epsilon y'_j}{2})}-\frac{exp(\frac{\epsilon^* y'_i}{2})}{\sum^k_{j=1}exp(\frac{\epsilon^* y'_j}{2})}|\\
&=\frac{\sum^k_{i=1}|exp(\frac{\epsilon y'_i}{2})\sum^k_{j=1}exp(\frac{\epsilon^* y'_j}{2})-exp(\frac{\epsilon^* y'_i}{2})\sum^k_{j=1}exp(\frac{\epsilon y'_j}{2})|}{\sum^k_{j=1}exp(\frac{\epsilon y'_j}{2})\cdot\sum^k_{j=1}exp(\frac{\epsilon^* y'_j}{2})}
\end{aligned}
\end{equation}

As $\sum^k_{j=1}exp(\frac{\epsilon y'_j}{2})\cdot\sum^k_{j=1}exp(\frac{\epsilon^* y'_j}{2})>0$, we focus only on $\sum^k_{i=1}|exp(\frac{\epsilon y'_i}{2})\sum^k_{j=1}exp(\frac{\epsilon^* y'_j}{2})-exp(\frac{\epsilon^* y'_i}{2})\sum^k_{j=1}exp(\frac{\epsilon y'_j}{2})|$. Because

\begin{equation}\nonumber
\begin{aligned}
&\sum^k_{i=1}|exp(\frac{\epsilon y'_i}{2})\sum^k_{j=1}exp(\frac{\epsilon^* y'_j}{2})-exp(\frac{\epsilon^* y'_i}{2})\sum^k_{j=1}exp(\frac{\epsilon y'_j}{2})|\\
&=\sum^k_{i=1}|\sum^k_{j=1}exp[\frac{1}{2}(\epsilon y'_i+\epsilon^* y'_j)]-\sum^k_{j=1}exp[\frac{1}{2}(\epsilon^* y'_i+\epsilon y'_j)]|,\\
\end{aligned}
\end{equation}

thus, evaluating $\sum^k_{i=1}|\sum^k_{j=1}exp[\frac{1}{2}(\epsilon y'_i+\epsilon^* y'_j)]-\sum^k_{j=1}exp[\frac{1}{2}(\epsilon^* y'_i+\epsilon y'_j)]|$ is equivalent to evaluating $\sum^k_{i=1}|\sum^k_{j=1}(\epsilon y'_i+\epsilon^* y'_j)-\sum^k_{j=1}(\epsilon^* y'_i+\epsilon y'_j)|$. We have 
\begin{equation}\nonumber
\begin{aligned}
&\sum^k_{i=1}|\sum^k_{j=1}(\epsilon y'_i+\epsilon^* y'_j)-\sum^k_{j=1}(\epsilon^* y'_i+\epsilon y'_j)|\\
&=\sum^k_{i=1}|k\epsilon y'_i+\sum^k_{j=1}\epsilon^* y'_j-(k\epsilon^* y'_i+\sum^k_{j=1}\epsilon y'_j)|\\
&=\sum^k_{i=1}|ky'_i(\epsilon-\epsilon^*)-(\epsilon-\epsilon^*)\sum^k_{j=1}y'_j|\\
&=\sum^k_{i=1}|(\epsilon-\epsilon^*)(ky'_i-\sum^k_{j=1}y'_j)|\\
&=|\epsilon-\epsilon^*|\sum^k_{i=1}|ky'_i-\sum^k_{j=1}y'_j|.\\
\end{aligned}
\end{equation}
Hence, $|\epsilon-\epsilon^*|$ determines the distance between $\mathbf{y}$ and $\mathbf{z}$, and the conclusion of this property is achieved.
\end{proof}

\subsection{Defending against membership inference attacks}\label{sub:defend inference}
In membership inference attacks, the attacker essentially exploits any overfitting of the target model, 
in that models often behave more confidently toward data 
on which they were trained versus data they are seeing for the first time \cite{Shokri17}. 
%As a result, compared with member data records, the overfitting target model is more uncertain to non-member data records \cite{Jia19}. 
Thus, the overarching aim of our defense method is to reduce the gap between 
the confidence score vectors of training set members versus non-members. 

To this end, we set a threshold $\tau$ to $0<\tau<1$, which is used to decide whether the target model is confident in an input data record. 
To explain, on the one hand, if $y_{max}>\tau$, the target model will appear confident in the input data record. 
Therefore, according to Property \ref{property4}, the defense method should reduce this confidence to confuse the attacker by setting $\epsilon<\epsilon^*$, 
where $\epsilon^*$ makes $\mathbf{z}=\mathbf{y}$. 
%The detailed discussion will be given in Property \ref{property4} in the next section.
The exact value of $\epsilon$ depends on the expected utility of the perturbed confidence score vector: $||\mathbf{z}-\mathbf{y}||_1$. 
%According to Theorem \ref{thm:DP}, if we want $||\mathbf{y}'-\mathbf{y}||_1\leq\delta$, 
%then the value of $\epsilon$ should be set in $(0,ln(1+\delta)]$. 
On the other hand, if $y_{max}\leq\tau$, the target model will appear less confident in the input data record. 
Therefore, the defense method should increase the confidence by setting $\epsilon>\epsilon^*$. 

Moreover, according to Property \ref{property2}, as the value of $\epsilon$ is always larger than $0$, 
the order of the scores in $\mathbf{y}$ will be preserved in $\mathbf{z}$ after perturbation, 
i.e., if $y_i$ is $y_{max}$ in $\mathbf{y}$, then $z_i$ is $z_{max}$ in $\mathbf{z}$. 
%More discussion will be given in Property \ref{property2} in the next section. 
Typically, the user of the target model will select the predicted class with the highest probability. 
Hence, this defense method does not affect user experience.

\subsection{Defending against model inversion attacks}\label{sub:defend inversion}
In model inversion attacks, the attacker trains an inversion model $G$ to approximate the inverse mapping of the target model $F$ \cite{Yang19}. 
The attack works due to the rich information contained in confidence vectors \cite{Fred15}. 
Thus, this defense method works to reduce the content of that information. 
One way to do this is to increase the distance between the confidence vector and its perturbed version. According to Property \ref{property5}, the distance between $\mathbf{y}$ and $\mathbf{z}$ is based on $|\epsilon-\epsilon^*|$. Hence, to widen the distance, if $\epsilon>\epsilon^*$, then $\epsilon$ should be set a large value. Otherwise, $\epsilon$ should be set a small value. 
In particular, when $\epsilon>\epsilon^*$, increasing the value of $\epsilon$ will increase the variance of the scores in $\mathbf{z}$. 
A high variance makes $\mathbf{z}$ appear as a confident prediction. 
Oppositely, when $\epsilon<\epsilon^*$, increasing the value of $\epsilon$ will decrease the variance of the scores in $\mathbf{z}$. 
A low variance makes $\mathbf{z}$ appear as an unconfident prediction. 
Thus, using a large or small value of $\epsilon$ depends on the original vector $\mathbf{y}$. 
If $\mathbf{y}$ is a confident prediction, $\epsilon$ should be set a small value 
which not only increases the distance between $\mathbf{y}$ and $\mathbf{z}$ but also makes $\mathbf{z}$ appear as an unconfident prediction. 
Otherwise, if $\mathbf{y}$ is an unconfident prediction, $\epsilon$ should be set a large value.

\section{Experiments}\label{sec:experiments}
\subsection{Experimental setup}
\subsubsection{Datasets}
The three datasets we chose for the experiments are broadly used in related studies. These are: 
\begin{itemize}
    \item \textbf{MNIST} \cite{LeCun98}, which consists of $70,000$ handwritten digit images in $10$ classes: 
    $0,1,2,3,4,5,6,7,8,9$. Each image has been resized to $32\times 32$.
    \item \textbf{Fashion-MNIST} \cite{Fashion} consisting of $70,000$ images across $10$ classes, 
    including T-shirt, trouser, pullover, dress, coat, sandal, shirt, sneaker, bag and ankle boot. 
    Again, each image has been resized to $32\times 32$.
    \item \textbf{CIFAR10} \cite{Krizhevsky14} with $60,000$ images across $10$ classes, 
    including airplane, automobile, bird, cat, deer, dog, horse, ship and truck, 
    also resized to $32\times 32$.
\end{itemize}

Table \ref{tab:data} presents the data allocation in our experiments. 
Note that the size of the attacker's training set is $60,000$ for MNIST and Fashion-MNIST and $50,000$ for CIFAR10. 
In the real world, it is infeasible for an attacker to collect a great many samples 
that share the same distribution with the training set of a target model. 
In the experiments, we attempt to build a strong attacker who can collect a large number of samples. 
Then, successfully defeating this strong attacker can prove the effectiveness of the proposed defense method.
\begin{table}[!ht]\scriptsize
\newcommand{\tabincell}[2]{\begin{tabular}{@{}#1@{}}#2\end{tabular}}
	\centering
	\caption{Data allocation}
\begin{tabular}{|c|c|c|c|c|} \hline
Dataset & $D^{\rm train}_{\rm target}$ & $D^{\rm test}_{\rm target}$ & $D^{\rm train}_{\rm attack}$ & $D^{\rm test}_{\rm attack}$ \\ \hline
MNIST & $60,000$ & $10,000$ & $60,000$ & $10,000$\\
Fashion-MNIST & $60,000$ & $10,000$ & $60,000$ & $10,000$\\
CIFAR10 & $50,000$ & $10,000$ & $50,000$ & $10,000$\\ \hline
\end{tabular}
	\label{tab:data}
\end{table}

\subsubsection{Target models}
We used the architecture proposed in \cite{Yang19} for the three datasets, 
which consists of three CNN blocks, two fully-connected layers and a softmax function. 
Each CNN block consists of a convolutional layer followed by a batch normalization layer, 
a max-pooling layer and a ReLU activation layer. 
The two fully-connected layers are added after the CNN blocks. 
Finally, the softmax function is added to the last layer to convert arbitrary neural signals 
into a valid confidence score vector $\mathbf{y}$.

%For the Purchase100 dataset, we use the same model as in \cite{Nasr18} 
%which is a 4-layer fully connected neural network with layer size: [1024,512,256,100]. 
%The activation function is Tanh.

\begin{table*}[]
\centering
\caption{Comprehensive results of the three defense methods}\label{tab:results}
\begin{tabular}{|c|c|c|c|c|c|c|c|c|c|}
\hline
\multirow{2}{8mm}{Dataset} & \multirow{2}{10mm}{Defense methods} & \multicolumn{3}{|c|}{Utility} & Model inversion & \multicolumn{2}{|c|}{Membership inference} & \multicolumn{2}{|c|}{Time overhead} \\ \cline{3-10}
& & Train acc. & Test acc. & Conf. dist. & Inversion error & ML-Leaks & NSH & Train (h) & Test (s) \\\hline
\multirow{7}{8mm}{MNIST} & No defense & $99.94\%$ & $99.63\%$ & $0$ & $0.935$ & $70.4\%$ & $72.3\%$ & $0$ & $0$ \\
                     & \textbf{DP-based} ($\epsilon=0.1$) & $\mathbf{99.94\%}$ & $\mathbf{99.63\%}$ & $\mathbf{0.948}$ & $\mathbf{0.924}$ & $\mathbf{49.8\%}$ & $\mathbf{50.3\%}$ & $\mathbf{0}$ & $\mathbf{6.96e-05}$ \\
                     & \textbf{DP-based} ($\epsilon=0.7$) & $\mathbf{99.94\%}$ & $\mathbf{99.63\%}$ & $\mathbf{0.783}$ & $\mathbf{0.925}$ & $\mathbf{50.4\%}$ & $\mathbf{50.9\%}$ & $\mathbf{0}$ & $\mathbf{6.95e-05}$ \\
                     & \textbf{DP-based} ($\epsilon=1.4$) & $\mathbf{99.94\%}$ & $\mathbf{99.63\%}$ & $\mathbf{0.535}$ & $\mathbf{0.927}$ & $\mathbf{51.1\%}$ & $\mathbf{51.6\%}$ & $\mathbf{0}$ & $\mathbf{6.97e-05}$ \\
                     & \textbf{DP-based} ($\epsilon=2.0$) & $\mathbf{99.94\%}$ & $\mathbf{99.63\%}$ & $\mathbf{0.328}$ & $\mathbf{0.928}$ & $\mathbf{52.0\%}$ & $\mathbf{52.5\%}$ & $\mathbf{0}$ & $\mathbf{6.95e-05}$ \\
                     & MemGuard & $99.94\%$ & $99.63\%$ & $0.392$ & $0.908$ & $57.2\%$ & $55.3\%$ & $1.02$ & $2.53$ \\
                     & Purification & $99.87\%$ & $99.55\%$ & $0.287$ & $0.925$ & $65.5\%$ & $68.2\%$ & $6.31$ & $1.22e-04$ \\\hline
\multirow{7}{8mm}{Fashion-MNIST} & No defense & $99.74\%$ & $92.73\%$ & $0$ & $0.708$ & $69.3\%$ & $71.3\%$ & $0$ & $0$ \\
                             & \textbf{DP-based} ($\epsilon=0.1$) & $\mathbf{99.74\%}$ & $\mathbf{92.73\%}$ & $\mathbf{0.948}$ & $\mathbf{0.691}$ & $\mathbf{49.7\%}$ & $\mathbf{50.1\%}$ & $\mathbf{0}$ & $\mathbf{6.91e-05}$ \\
                             & \textbf{DP-based} ($\epsilon=0.7$) & $\mathbf{99.74\%}$ & $\mathbf{92.73\%}$ & $\mathbf{0.774}$ & $\mathbf{0.694}$ & $\mathbf{50.5\%}$ & $\mathbf{50.8\%}$ & $\mathbf{0}$ & $\mathbf{6.91e-05}$ \\
                             & \textbf{DP-based} ($\epsilon=1.4$) & $\mathbf{99.74\%}$ & $\mathbf{92.73\%}$ & $\mathbf{0.514}$ & $\mathbf{0.695}$ & $\mathbf{51.2\%}$ & $\mathbf{51.7\%}$ & $\mathbf{0}$ & $\mathbf{6.90e-05}$ \\
                             & \textbf{DP-based} ($\epsilon=2.0$) & $\mathbf{99.74\%}$ & $\mathbf{92.73\%}$ & $\mathbf{0.316}$ & $\mathbf{0.697}$ & $\mathbf{51.9\%}$ & $\mathbf{52.5\%}$ & $\mathbf{0}$ & $\mathbf{6.93e-05}$ \\
                             & MemGuard & $99.74\%$ & $92.73\%$ & $0.442$ & $0.697$ & $55.2\%$ & $54.2\%$ & $1.1$ & $2.62$ \\
                             & Purification & $99.68\%$ & $92.65\%$ & $0.301$ & $0.693$ & $65.1\%$ & $67.1\%$ & $6.40$ & $1.22e-04$ \\\hline
\multirow{7}{8mm}{CIFAR10} & No defense & $99.14\%$ & $81.63\%$ & $0$ & $0.392$ & $65.6\%$ & $63.2\%$ & $0$ & $0$ \\
                     & \textbf{DP-based} ($\epsilon=0.1$) & $\mathbf{99.14\%}$ & $\mathbf{81.63\%}$ & $\mathbf{0.947}$ & $\mathbf{0.426}$ & $\mathbf{50.1\%}$ & $\mathbf{49.3\%}$ & $\mathbf{0}$ & $\mathbf{7.41e-05}$ \\
                     & \textbf{DP-based} ($\epsilon=0.7$) & $\mathbf{99.14\%}$ & $\mathbf{81.63\%}$ & $\mathbf{0.787}$ & $\mathbf{0.428}$ & $\mathbf{50.8\%}$ & $\mathbf{49.9\%}$ & $\mathbf{0}$ & $\mathbf{7.40e-05}$ \\
                     & \textbf{DP-based} ($\epsilon=1.4$) & $\mathbf{99.14\%}$ & $\mathbf{81.63\%}$ & $\mathbf{0.495}$ & $\mathbf{0.429}$ & $\mathbf{51.5\%}$ & $\mathbf{50.6\%}$ & $\mathbf{0}$ & $\mathbf{7.39e-05}$ \\
                     & \textbf{DP-based} ($\epsilon=2.0$) & $\mathbf{99.14\%}$ & $\mathbf{81.63\%}$ & $\mathbf{0.305}$ & $\mathbf{0.502}$ & $\mathbf{52.3\%}$ & $\mathbf{51.5\%}$ & $\mathbf{0}$ & $\mathbf{7.42e-05}$ \\
                     & MemGuard & $99.14\%$ & $81.63\%$ & $0.327$ & $0.431$ & $53.9\%$ & $52.9\%$ & $5.1$ & $2.7$ \\
                     & Purification & $99.09\%$ & $81.56\%$ & $0.284$ & $0.403$ & $61.2\%$ & $58.0\%$ & $8.24$ & $1.24e-04$ \\\hline                             
\end{tabular}
\end{table*}

\subsubsection{Attack models}
We used two different attack models for the membership inference attacks. 

\noindent\textbf{ML-leak attack} \cite{Salem19}. This is a confidence-based membership inference attack. 
The attacker has no knowledge of the $D^{\rm train}_{\rm attack}$ and $D^{\rm test}_{\rm attack}$ membership labels. 
Thus, a shadow model has to be trained to replicate the target model; 
then the attack model must be trained based on the confidence scores of the shadow model. 
To guarantee the strongest attack, the shadow model should have the same architecture as the target model. 
We use the same architecture as in \cite{Salem19} for the attack model 
which is a multi-layer perceptron with a 64-unit hidden layer and a sigmoid output layer.
    
\noindent\textbf{NSH attack} \cite{Nasr18}. This is a combined 
confidence/label-based membership inference attack. 
The attacker has knowledge of the $D^{\rm train}_{\rm attack}$ and $D^{\rm test}_{\rm attack}$ membership labels. 
Thus, no shadow model is needed. 
The adversary can simply directly query the target model to receive the confidence score vectors. 
The architecture is the same as in \cite{Nasr18} 
which consists of three neural networks. 
The first has the layers of size: [100,1024,512,64] 
and takes confidence score vectors as input. 
The second has the layers of size: [100,512,64] 
and takes labels as input. 
The third network has the layers of size: [256,64,1] 
and takes the outputs of the first and second networks as input.

For the model inversion attacks, we adopted the model proposed in \cite{Yang19}. 

\noindent\textbf{Adversarial model inversion attack} \cite{Yang19}. 
The adversary trains an inversion model to infer reconstruction of the input data record. 
We use the same inversion model architecture as in \cite{Yang19} 
which consists of four transposed CNN blocks. 
The first three blocks each has a transposed convolutional layer 
followed by a sigmoid activation function that converts neural signals into real values in $[0,1]$.

\subsubsection{Comparison defense methods} 
For comparison, we chose two existing defense methods 
that have been experimentally proven as state-of-the-art \cite{Yang20}. 
These are the \textbf{MemGuard} method \cite{Jia19} for the membership inference attacks 
and the \textbf{Purification} method \cite{Yang20} for the model inversion attacks.

\noindent\textbf{MemGuard} \cite{Jia19}. The defense model consists of three hidden layers: [256,128,64].
It uses ReLU in the hidden layers and sigmoid in the output layer. 

\noindent\textbf{Purification} \cite{Yang20}. The defense model used is an autoencoder 
with the layers of size [10,7,4,7,10]. % for CIFAR10 and MNIST, 
%and [100,50,20,10,20,50,100] for Purchase100. 
Every hidden layer uses a ReLU activation function and batch normalization.

\subsubsection{Evaluation metrics}
We used five metrics to evaluate the performance and efficiency of these defense methods following the specifications outlined in \cite{Yang20}. 
A brief description of each follows. 

\noindent\textbf{Classification accuracy}. This metric demonstrates the performance of target models on classification tasks.
It is measured on the training set $D^{\rm train}_{\rm target}$ 
and test set $D^{\rm test}_{\rm target}$ of target models. 

\noindent\textbf{Confidence score distortion}. This metric shows the utility of the perturbed confidence score vectors. 
As analyzed in \cite{Bozkir20}, the utility is measured by computing the $l_2$ norm of the distance 
between an original confidence score vector, predicted by a target model, 
and a perturbed confidence score vector, computed using the defense method.

\noindent\textbf{Membership inference accuracy}. This metric shows the classification accuracy of attack models 
in predicting the membership of input data records. 
It is measured on $D^{\rm train}_{\rm target}-D^{\rm train}_{\rm attack}$, i.e., members, 
and $D^{\rm test}_{\rm target}-D^{\rm test}_{\rm attack}$, i.e., non-members. 

\noindent\textbf{Inversion error}. This metric shows the reconstruction accuracy of the attack model 
in reconstructing the input data records. 
It is measured by computing the mean squared error between the original input data record 
and the reconstructed data record.
It is measured on datasets $D^{\rm train}_{\rm target}$ and $D^{\rm test}_{\rm target}$.

\noindent\textbf{Time overhead}. This metric indicates the efficiency of the defense methods. 
It is measured by reporting the extra time consumed by applying defense methods. 
The time overhead includes both the training time of any models introduced by these defense methods
and the test time when using these models.

\subsection{Experimental results}
\subsubsection{Comparison with existing defense methods}
The full results of the comparisons appear in Table \ref{tab:results}.  
%where in our DP-based method, $\epsilon$ is set to $0.1$. 
As shown, our method significantly reduced the membership inference accuracy 
with no classification accuracy loss and $0$ training time. 
By comparison, the other defense methods had a very high training time. 
The test time overhead of our method, i.e., perturbing the confidence score vectors, 
was also much lower than the other methods. 
MemGuard's test time is consumed by solving an optimization problem, 
while for the Purification method, 
it is incurred in computing a forward pass of the defense model. 

After applying our method, the membership inference accuracy drops 
about $20\%$ on MNIST and Fashion-MNIST datasets, and about $15\%$ on CIFAR10 dataset. 
Thus, our method renders the results of a membership inference attack down to little more than a random guess for the attacker, 
which is better than the other two methods. 
Additionally, our method yielded a larger distortion in confidence scores than the other two methods. %, when $\epsilon$ is set to $0.1$. 
We reason this is because those two methods treat perturbing the confidence scores as an optimizing problem. 
Our method is also capable of less distortion, simply by setting a larger $\epsilon$ value. However, increase in the value of $\epsilon$ will incur the rise of membership inference attack success rates. Thus, here is a trade-off between the confidence score distortion and the membership inference attack success rate. During the experiments, we found that confidence score distortion was not a critical factor in classification accuracy. This is due to the fact that a classifier typically selects the label with the highest score as the output. Hence, any defense method can guarantee that the classification accuracy level will be maintained, as long as there is also a guarantee that the scores in the perturbed vector will be in the same order as the original vector. Therefore, we can focus on tuning $\epsilon$ value to reduce attack success rates. 
%This is due to the fact that a classifier typically selects 
%the label with the highest score as the output. 

%\begin{table}[!ht]\scriptsize
%\newcommand{\tabincell}[2]{\begin{tabular}{@{}#1@{}}#2\end{tabular}}
%	\centering
%	\caption{Confidence score distortion of our DP-based method with different $\epsilon$ values}
%\begin{tabular}{|c|c|c|c|c|} \hline
%$\epsilon$ & $0.1$ & $0.7$ & $1.4$ & $2.0$ \\ \hline
%MNIST & $0.948$ & $0.783$ & $0.535$ & $0.328$\\
%Fashion-MNIST & $0.948$ & $0.774$ & $0.514$ & $0.316$\\
%CIFAR10 & $0.947$ & $0.787$ & $0.495$ & $0.305$\\ \hline
%\end{tabular}
%	\label{tab:distortion}
%\end{table}

\begin{figure}[ht]
\centering
	\includegraphics[scale=0.32]{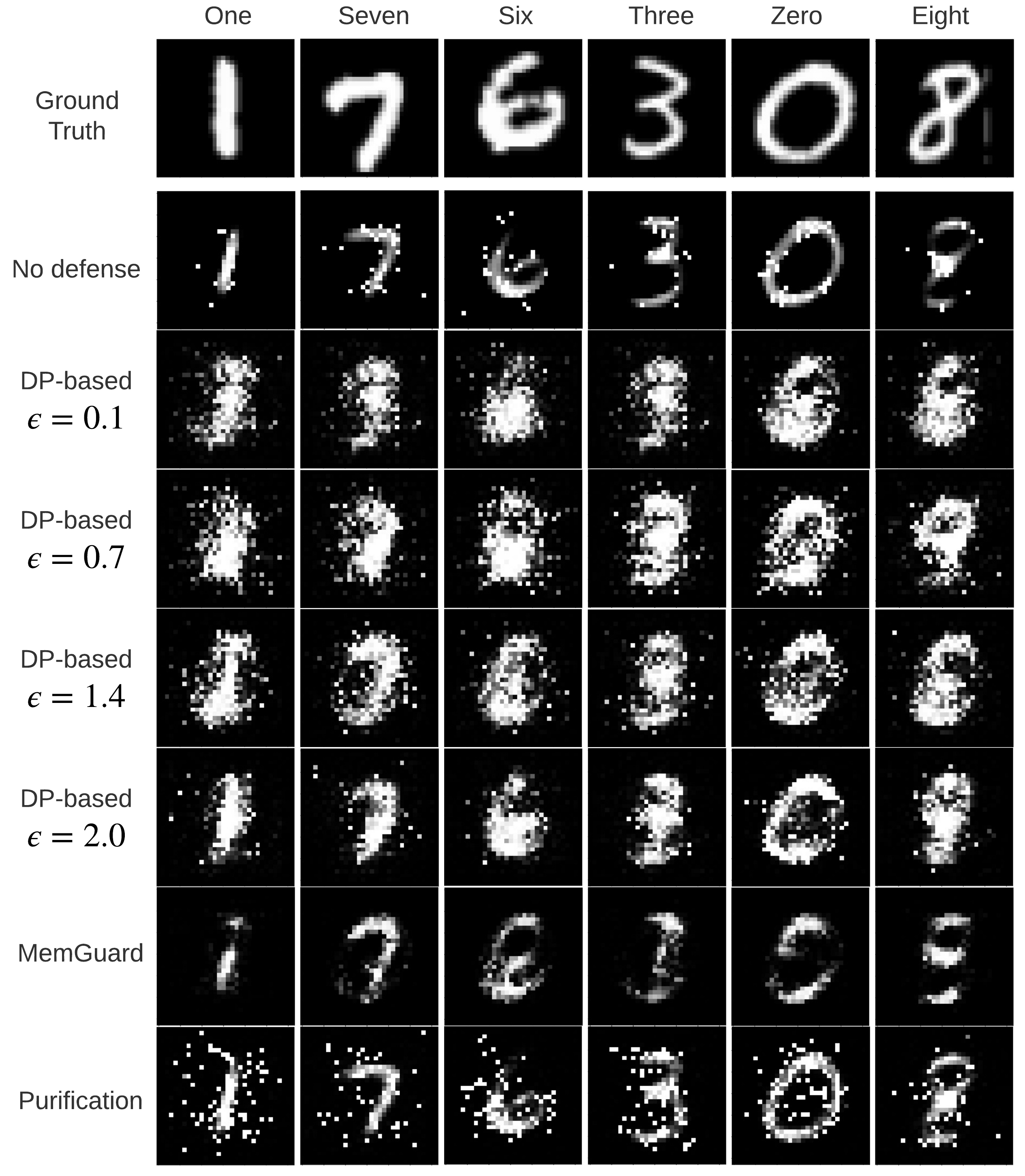}
	\caption{Defenses against the model inversion attacks on MNIST. 
	Rows 3-6 show the results for our DP-based method with different $\epsilon$ values.}
	\label{fig:MNIST}
\end{figure}

\begin{figure}[ht]
\centering
	\includegraphics[scale=0.32]{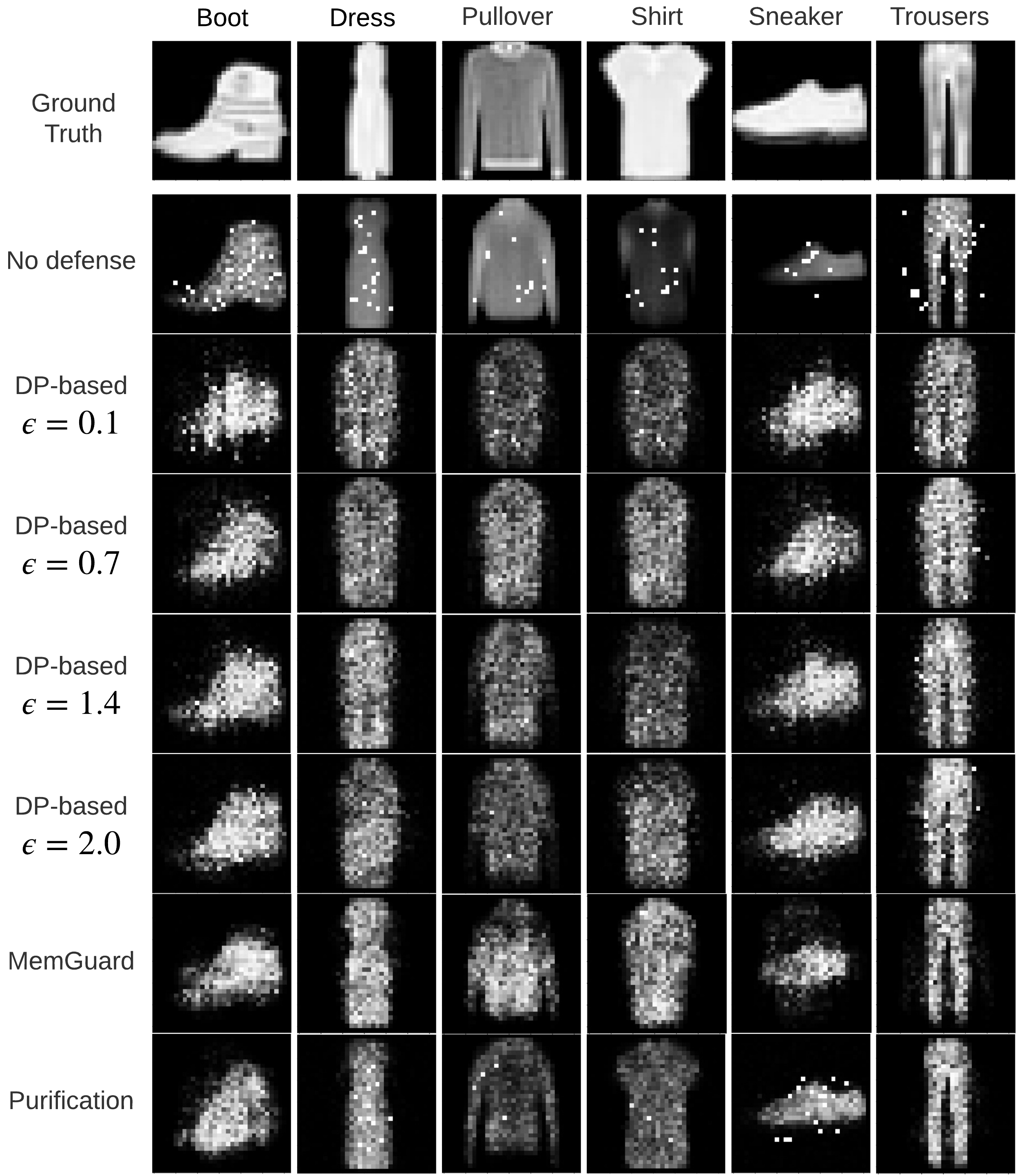}
	\caption{Defenses against the model inversion attacks on Fashion-MNIST. 
	Rows 3-6 show the results for our DP-based method with different $\epsilon$ values.}
	\label{fig:FMNIST}
\end{figure}

The results for model inversion attacks are shown in Figs. \ref{fig:MNIST}, \ref{fig:FMNIST} and \ref{fig:CIFAR10}. 
We drew three interesting findings from these experiments. 
First, the reconstructed images reveal the average features of one class of images. 
This means that the images belonging to one class have a very similar reconstructed image. 
The second finding is that the quality of reconstructed images depends heavily on 
the color and background of the original images. 
A grey-scale image with no background usually gives rise to a much better reconstructed version 
than a colorful image with a very rich background. 
As shown in the second row in Fig. \ref{fig:CIFAR10}, even without defense, 
the model inversion attack method \cite{Yang19} cannot precisely reconstruct images in CIFAR10. 
Then, as shown in the third row, using our defense method can make the attack results even worse. 
The third finding is that inversion error is not critical to the quality of the reconstructed images. 
This is because the aim of model inversion attack is usually for human perception. 
For example, slightly rotating or adding a small amount of noise 
to an image has a negligible impact on human perception 
but may induce huge mean squared errors, i.e., inversion errors. 
Thus, even if a reconstructed image has a huge inversion error compared with the original one, 
it may still be recognizable to a person. 
This finding also explains the converse that some reconstructed images 
have very bad quality despite small inversion errors. 
The use of an image-specific evaluation metric, 
e.g., structural similarity index measure (SSIM) \cite{Wang04}, is left to future work.

\begin{figure}[ht]
\centering
	\includegraphics[scale=0.49]{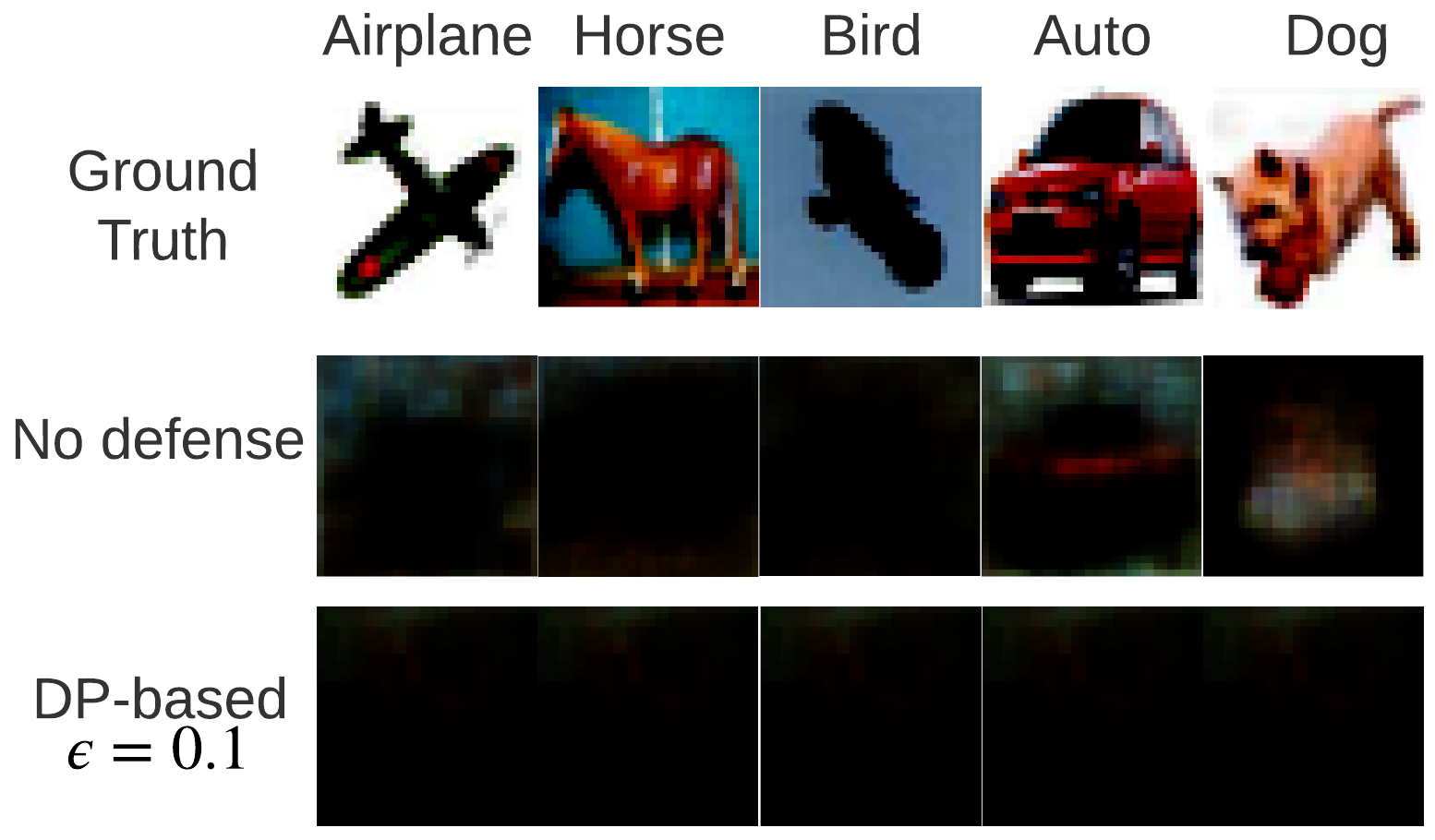}
	\caption{Defenses against the model inversion attacks on CIFAR10. 
	Row 3 shows the results for our DP-based method with $\epsilon=0.1$.}
	\label{fig:CIFAR10}
\end{figure}

From Figs. \ref{fig:MNIST} and \ref{fig:FMNIST}, we can see that, 
without any defense mechanism, the attacker can infer very accurate reconstructions of the images. 
However, with a defense, the inversion results become vague. 
Our method ``averages'' the images 
using differential privacy making them look more the same by removing 
useful information from the confidence score vectors. 
Moreover, as the $\epsilon$ value increases, the reconstructed images become clearer. 
This can be explained by the fact that when $\epsilon<\epsilon^*$, 
a larger $\epsilon$ value introduces less perturbation 
which implies less information removal. 

Notably, MemGuard achieved very good results, 
even though it is not designed to defend against model inversion attacks. 
It also removed the ``bright points'' from the reconstructed images. 
These ``bright points'' represent the average features of one class. 
%and the MemGuard method can remove these features. 
Moreover, the Purification method also did a commendable obfuscation job.  
%by obfuscating the reconstructed images.

\begin{figure}[ht]
	\begin{minipage}{1\textwidth}
   \subfigure[\scriptsize{Classification accuracy in MNIST}]{
    \includegraphics[scale=0.195]{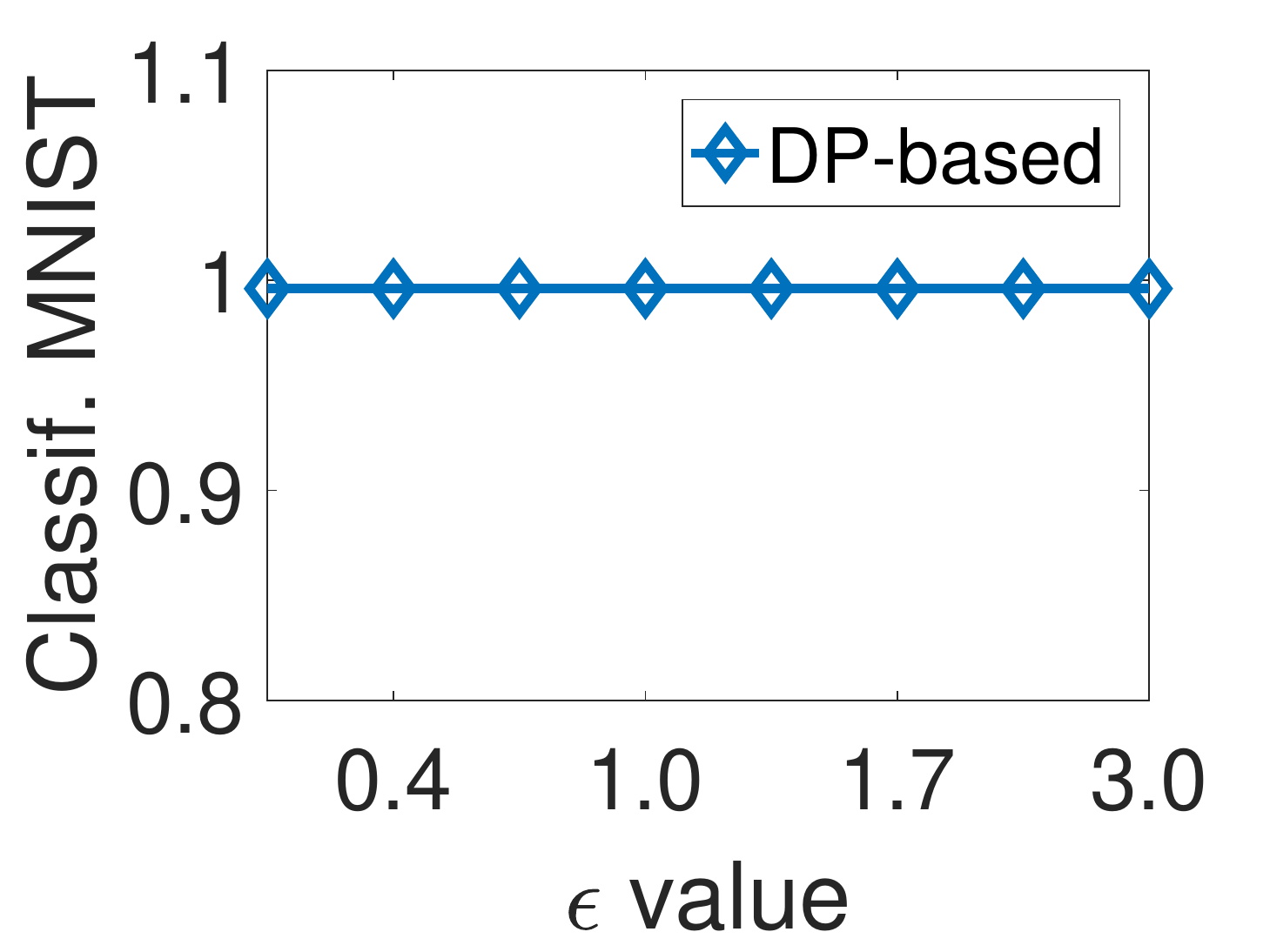}%width=0.32\textwidth, height=2.5cm
			\label{fig:CA-MNIST}}
	\subfigure[\scriptsize{Classification accuracy in Fashion-MNIST}]{
    \includegraphics[scale=0.195]{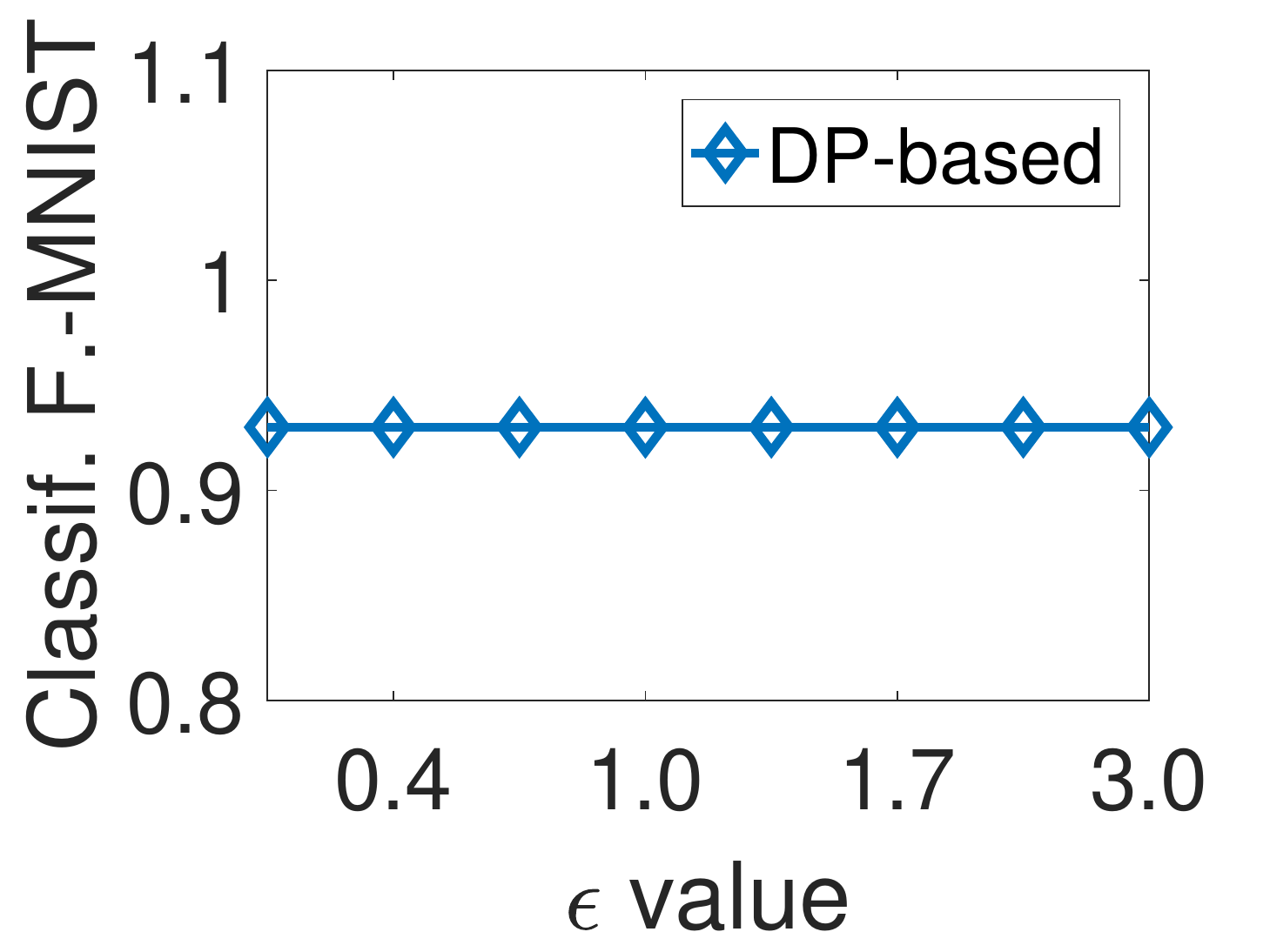}%width=0.32\textwidth, height=2.5cm
			\label{fig:CA-FMNIST}}%\\[2ex]
	\subfigure[\scriptsize{Classification accuracy in CIFAR10}]{
    \includegraphics[scale=0.195]{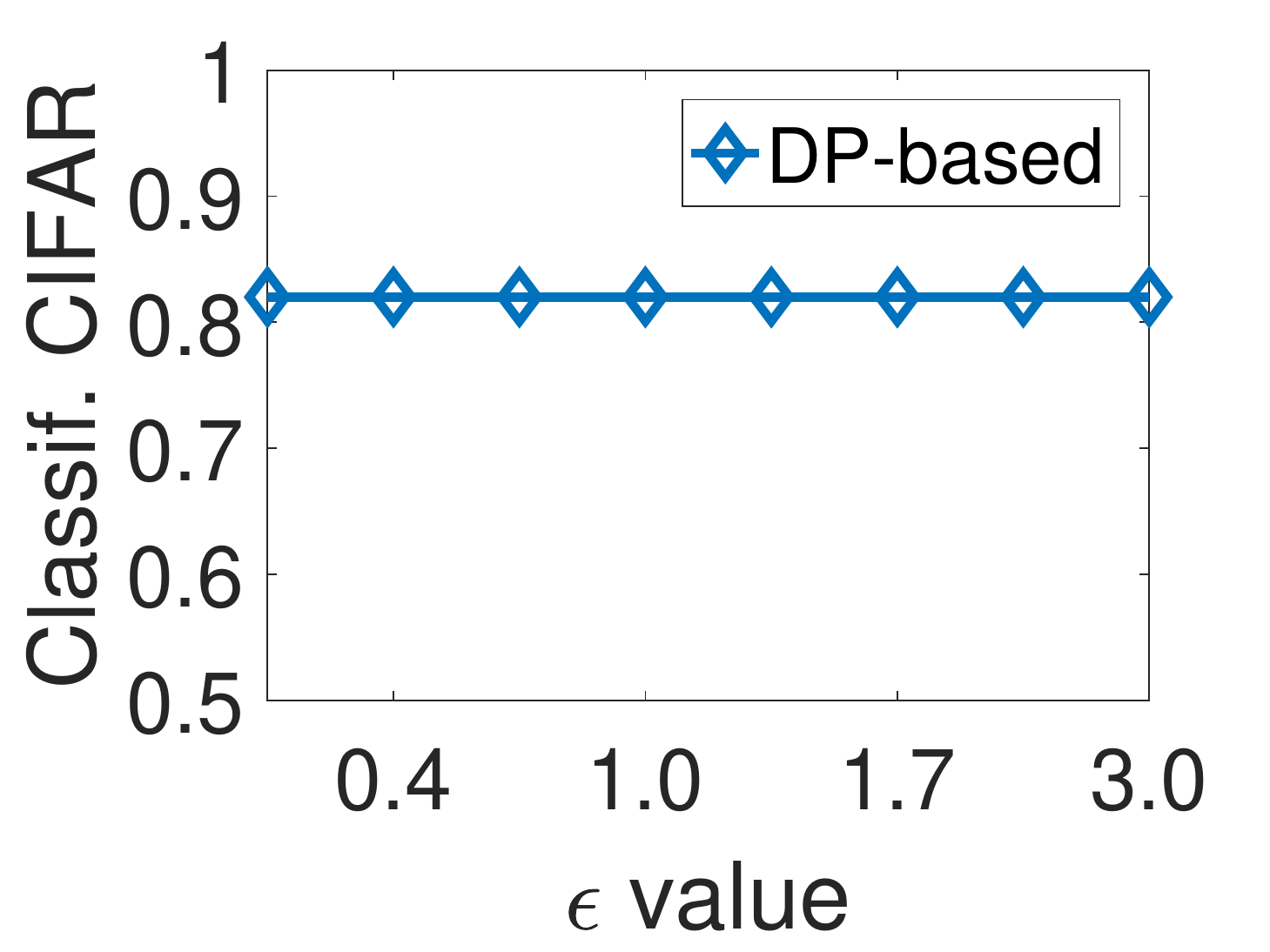}%width=0.32\textwidth, height=2.5cm
			\label{fig:CA-CIFAR}}
    \end{minipage}
	\caption{Classification accuracy with varying $\epsilon$ values}
	\label{fig:CA}
\end{figure}

\begin{figure}[ht]
	\begin{minipage}{1\textwidth}
   \subfigure[\scriptsize{Confidence score distortion in MNIST}]{
    \includegraphics[scale=0.195]{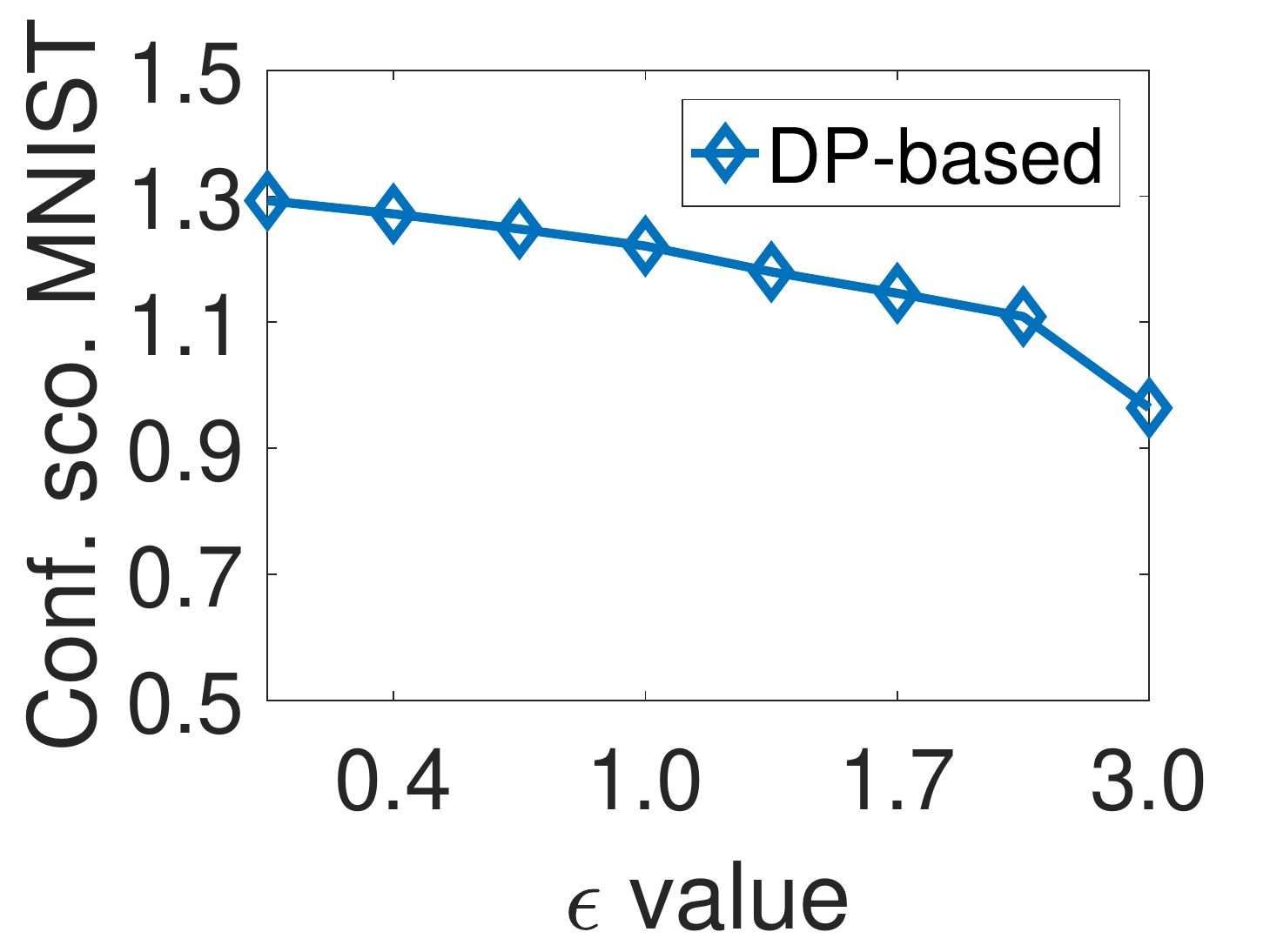}%width=0.32\textwidth, height=2.5cm
			\label{fig:CSD-MNIST}}
	\subfigure[\scriptsize{Confidence score distortion in Fashion-MNIST}]{
    \includegraphics[scale=0.195]{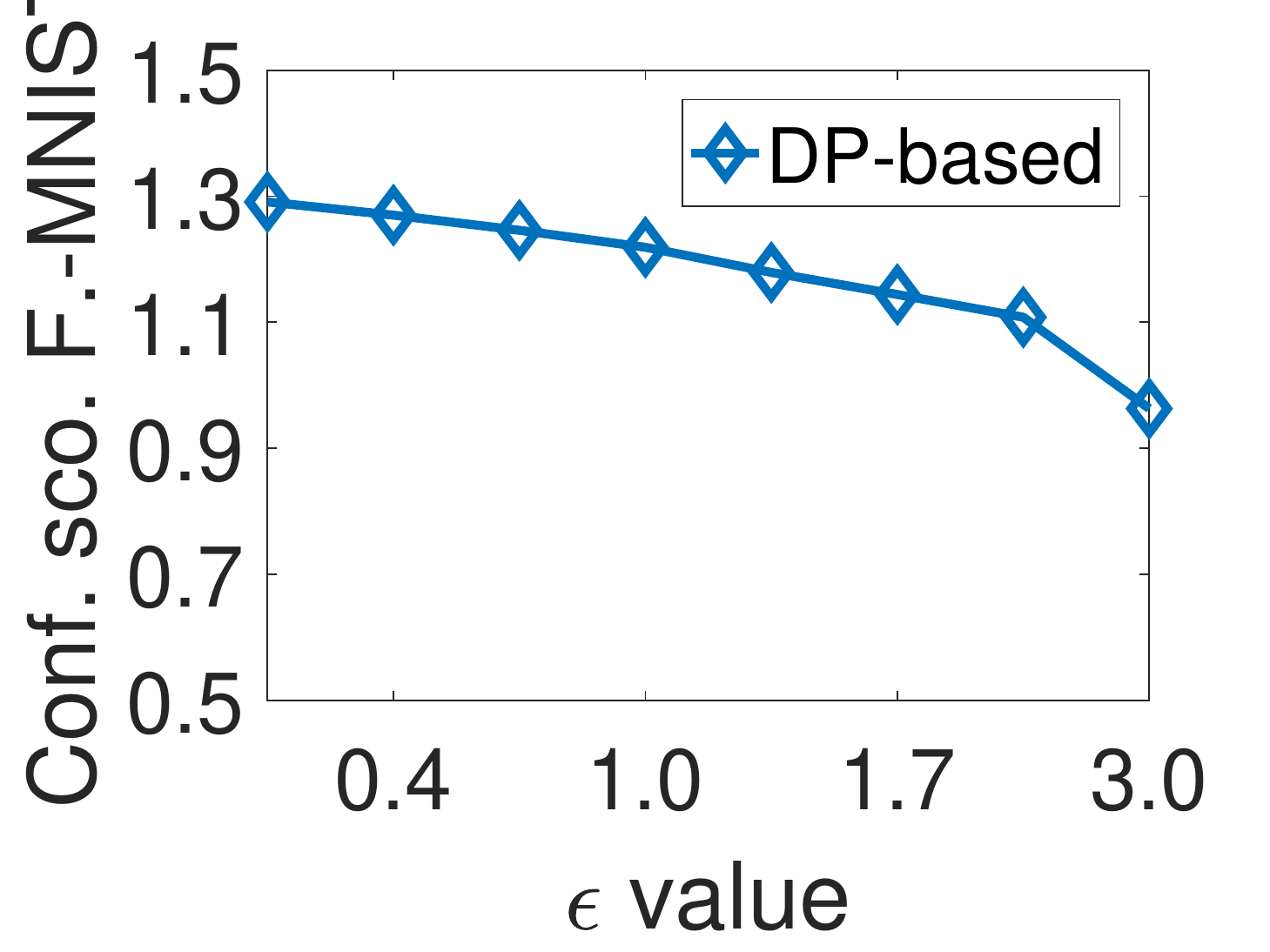}%width=0.32\textwidth, height=2.5cm
			\label{fig:CSD-FMNIST}}%\\[2ex]
	\subfigure[\scriptsize{Confidence score distortion in CIFAR10}]{
    \includegraphics[scale=0.195]{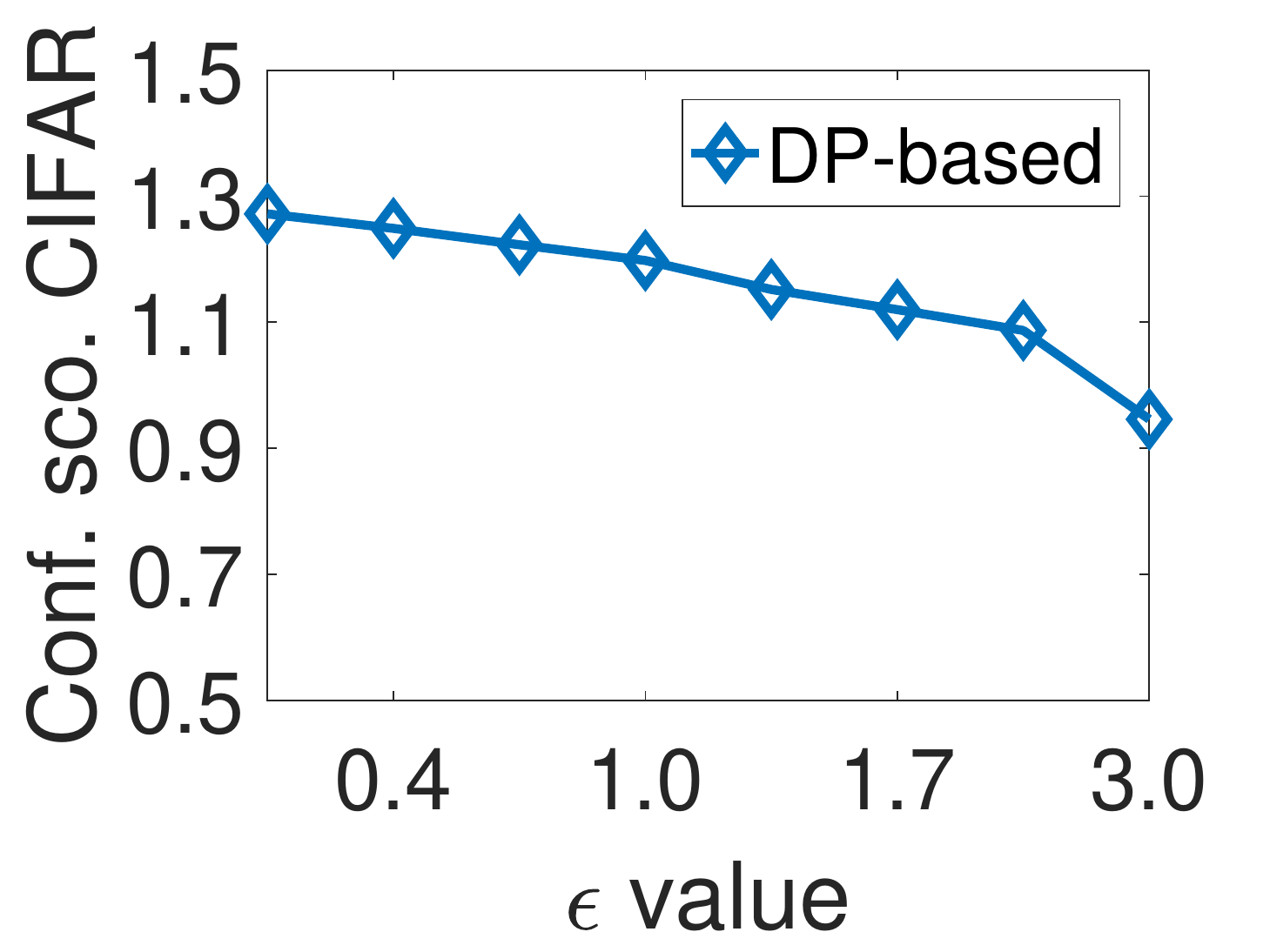}%width=0.32\textwidth, height=2.5cm
			\label{fig:CSD-CIFAR}}
    \end{minipage}
	\caption{Confidence score distortion with varying $\epsilon$ values}
	\label{fig:CSD}
\end{figure}

\begin{figure}[ht]
	\begin{minipage}{1\textwidth}
   \subfigure[\scriptsize{Membership inference accuracy in MNIST}]{
    \includegraphics[scale=0.195]{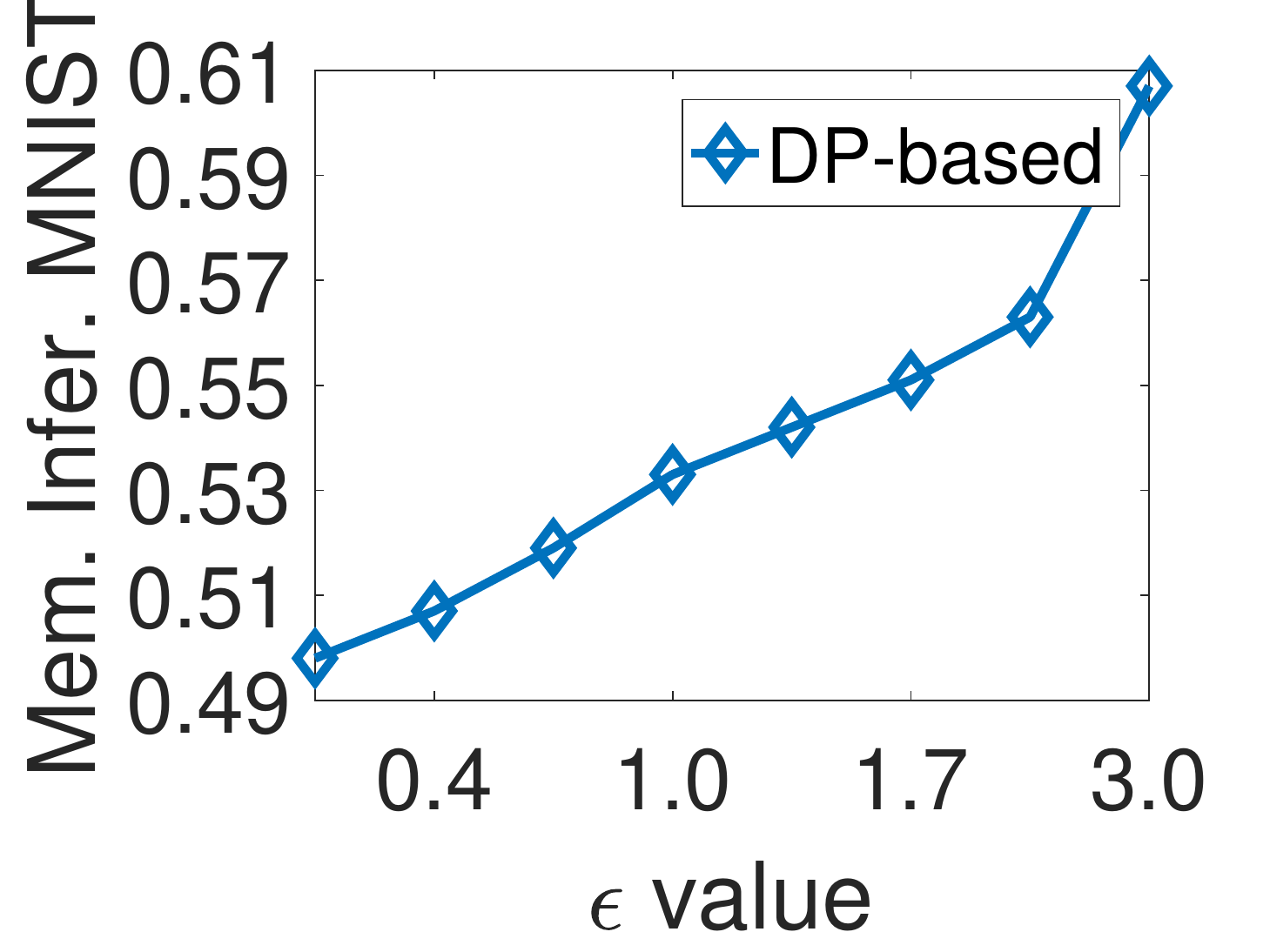}%width=0.32\textwidth, height=2.5cm
			\label{fig:MIA-MNIST}}
	\subfigure[\scriptsize{Membership inference accuracy in Fashion-MNIST}]{
    \includegraphics[scale=0.195]{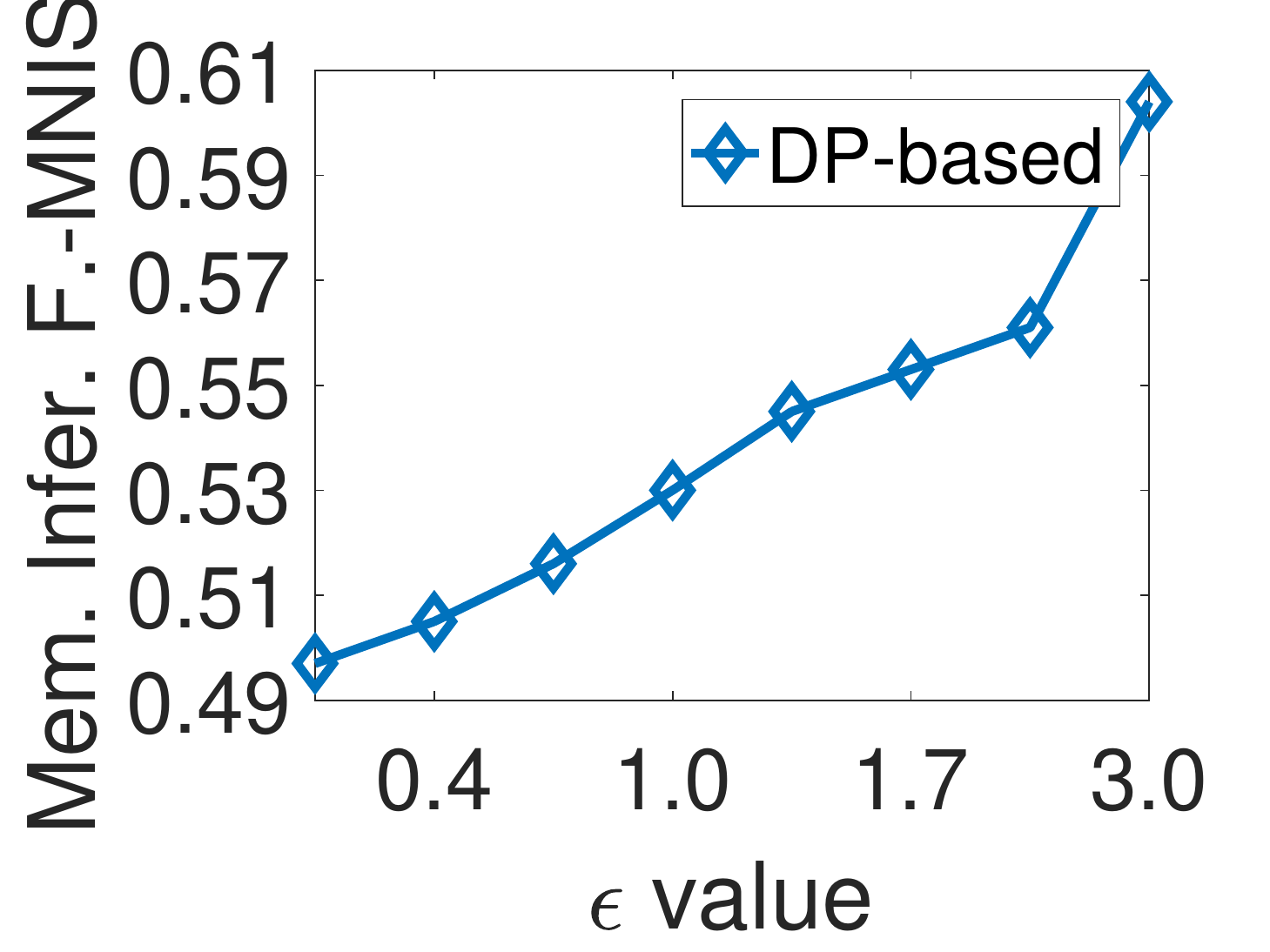}%width=0.32\textwidth, height=2.5cm
			\label{fig:MIA-FMNIST}}%\\[2ex]
	\subfigure[\scriptsize{Membership inference accuracy in CIFAR10}]{
    \includegraphics[scale=0.195]{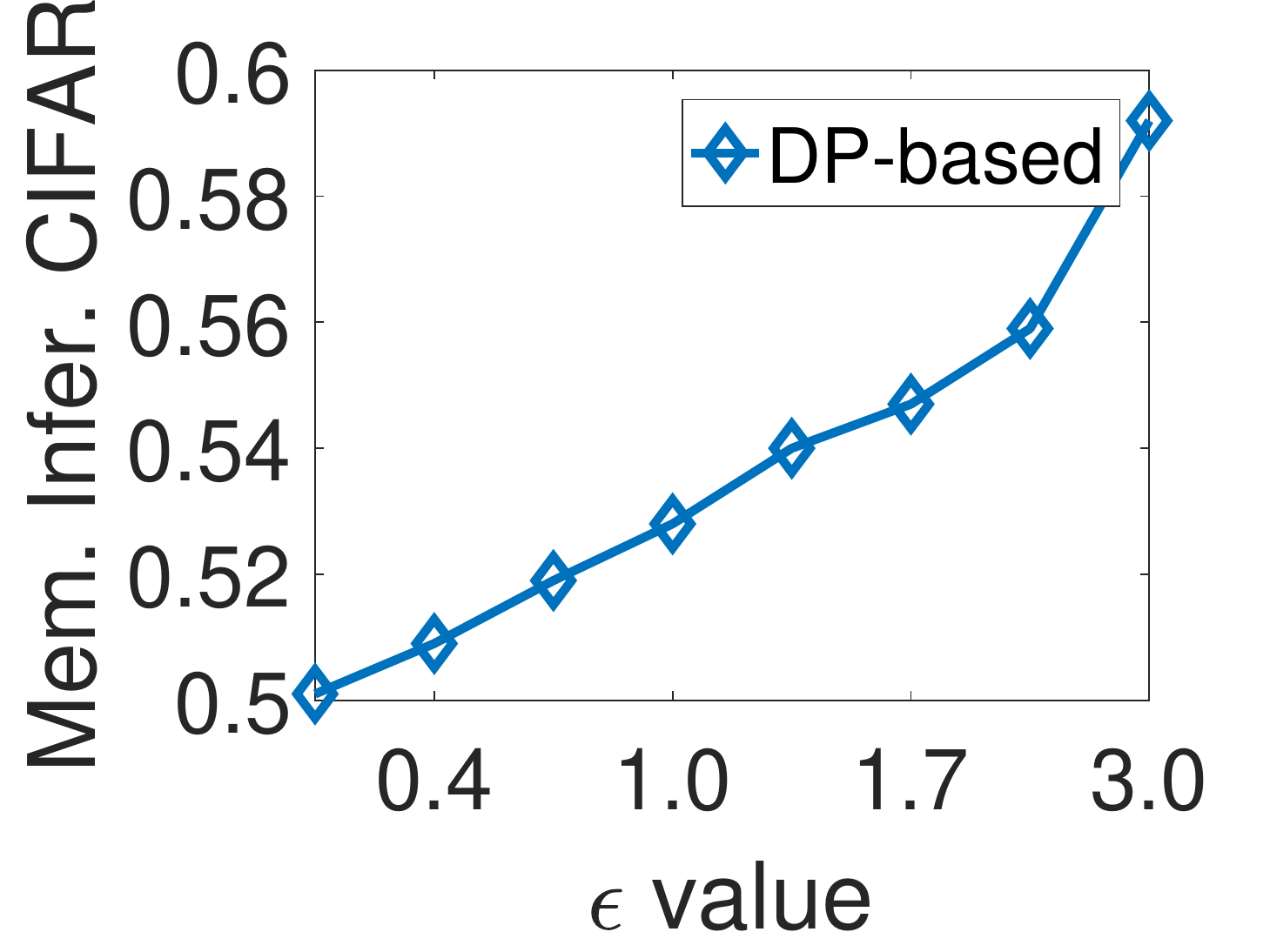}%width=0.32\textwidth, height=2.5cm
			\label{fig:MIA-CIFAR}}
    \end{minipage}
	\caption{Membership inference accuracy with varying $\epsilon$ values}
	\label{fig:MIA}
\end{figure}

\begin{figure}[ht]
	\begin{minipage}{1\textwidth}
   \subfigure[\scriptsize{Inversion error in MNIST}]{
    \includegraphics[scale=0.195]{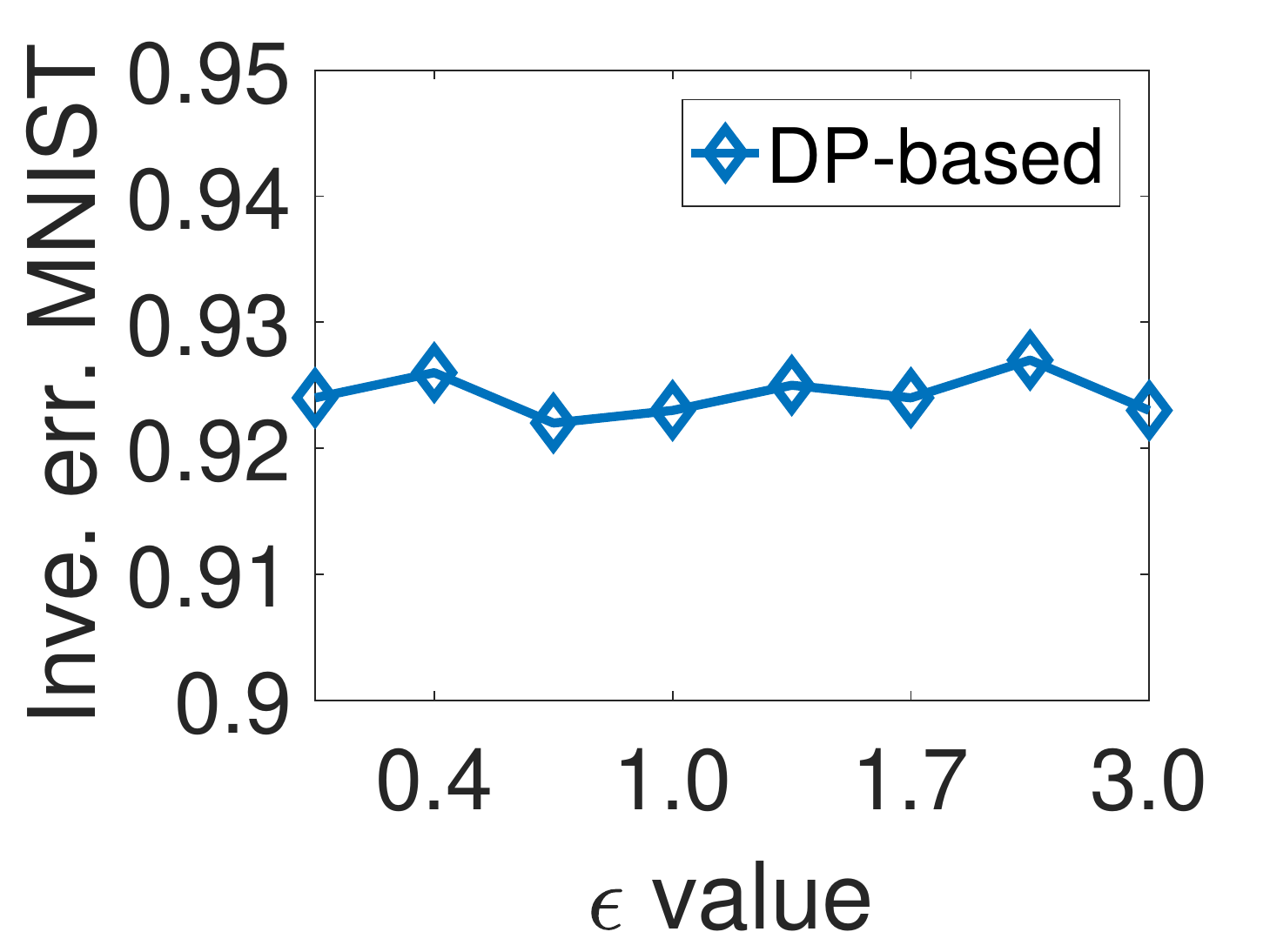}%width=0.32\textwidth, height=2.5cm
			\label{fig:IE-MNIST}}
	\subfigure[\scriptsize{Inversion error in Fashion-MNIST}]{
    \includegraphics[scale=0.195]{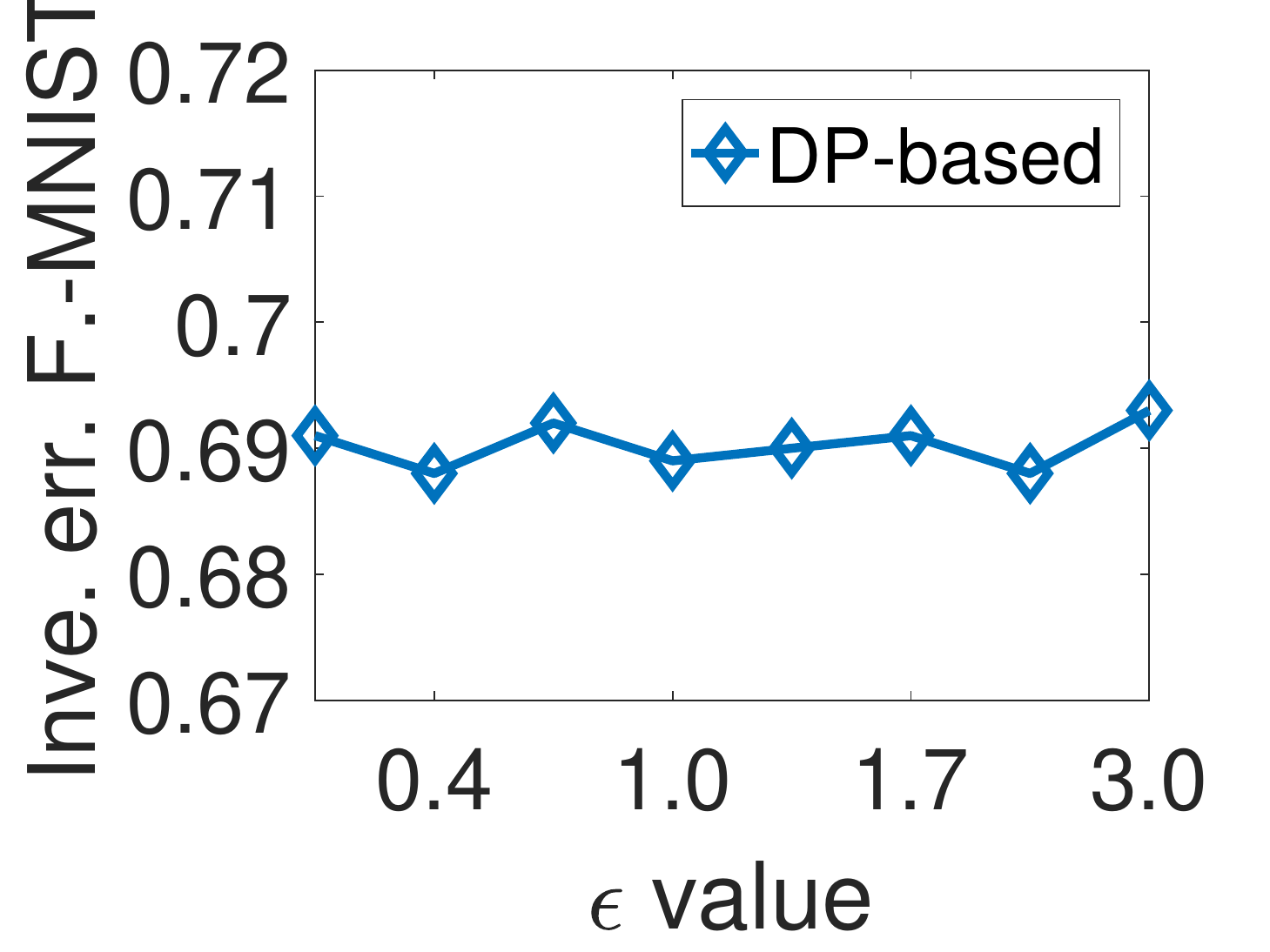}%width=0.32\textwidth, height=2.5cm
			\label{fig:IE-FMNIST}}%\\[2ex]
	\subfigure[\scriptsize{Inversion error in CIFAR10}]{
    \includegraphics[scale=0.195]{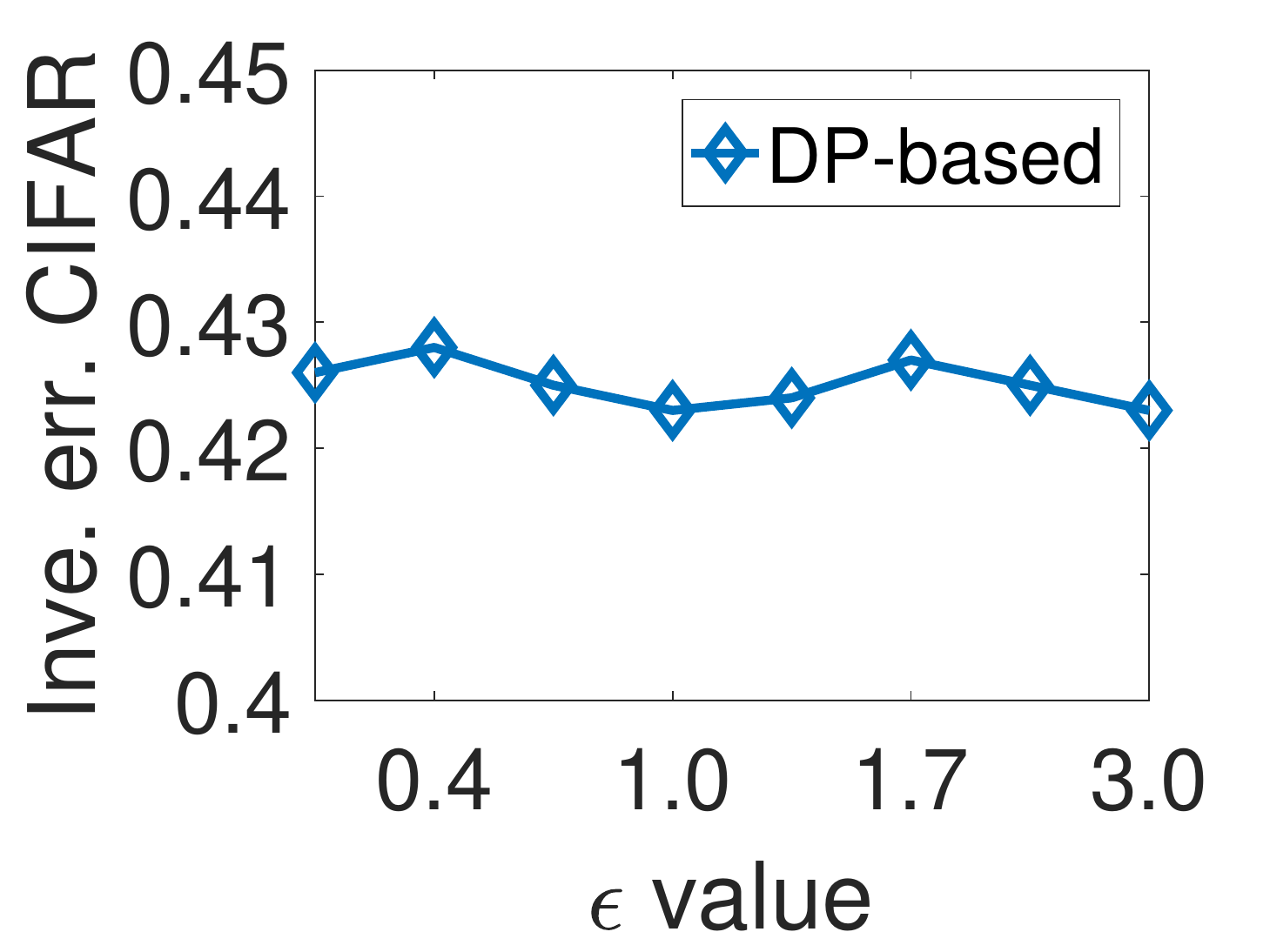}%width=0.32\textwidth, height=2.5cm
			\label{fig:IE-CIFAR}}
    \end{minipage}
	\caption{Inversion error with varying $\epsilon$ values}
	\label{fig:IE}
\end{figure}

\subsubsection{The impact of different $\epsilon$ values on our method}
Figs. \ref{fig:CA}, \ref{fig:CSD}, \ref{fig:MIA} and \ref{fig:IE} 
demonstrate the impact of varying $\epsilon$ values 
across the five metrics on our method. 
We can see that as the $\epsilon$ value increases, 
classification accuracy remains the same, 
confidence score distortion decreases, 
membership inference accuracy rises, 
and inversion error stays mostly steady. 

In terms of classification accuracy, as explained above, 
since our method preserves the order of the scores in a confidence vector, 
classification accuracy is not affected by the value of $\epsilon$. 
For confidence score distortion and membership inference accuracy, 
as analyzed in Section \ref{sec:analysis}, when $\epsilon<\epsilon^*$, 
a larger $\epsilon$ value incurs a smaller confidence score distortion 
which leads to higher membership inference accuracy. 
Hence, the confidence score distortion does not affect the classification accuracy 
but it does have a huge impact on membership inference accuracy. 
Finally, the inversion error is not much affected by the value of $\epsilon$. 
This is because the inversion error is used as the loss function to train the attack model. 
Therefore, as long as the attack model converges, 
the inversion error will also converge to a relatively narrow range. 
It should be noted, however, that although different $\epsilon$ values yield almost the same inversion error, 
they can lead to very different model inversion results, 
once more demonstrating that inversion errors are 
not a critical determinant in the model inversion results. 

\section{Conclusion}\label{sec:conclusion}
In this paper, we proposed a differentially private and time-efficient defense method against 
both membership inference attacks and model inversion attacks. 
Our strategy is to use an exponential mechanism to modify and normalize 
the confidence score vectors to confuse the attacker's model. 
The experimental results show that this approach outperforms existing defense methods 
in various respects, especially, in terms of maintaining classification accuracy loss and not incurring training overhead. 
In future, we plan to extend our method to handle the attacks 
that only make use of labels \cite{Choo20,Li21}. 
A possible method is to modify the output label using a differential privacy mechanism, e.g., exponential mechanism. 
This, certainly, will introduce a classification error. 
However, as every target model has an intrinsic classification error when classifying a dataset, 
we need only to control the introduced classification error smaller than the intrinsic classification error 
by properly tuning the privacy budge $\epsilon$.
%The current defense methods, which rely on confidence score vectors, may not defeat label-only attacks. 
Another future work is using image-specific evaluation metrics in our experiments, 
e.g., SSIM, to measure the quality of reconstructed images. %more precisely than the metric, inversion error.

%The encoder maps a confidence score vector to a latent representation. 
%The decoder then maps the representation to a reconstruction.


\section*{Acknowledgment}
This paper is supported by an ARC project, DP190100981, from the Australian Research Council, Australia. 
We also much appreciate the PhD candidate, Shuai Zhou, for his experimental support.

\bibliographystyle{IEEEtran}
{\small \bibliography{references}}

\begin{thebibliography}{10}
\providecommand{\url}[1]{#1}
\csname url@samestyle\endcsname
\providecommand{\newblock}{\relax}
\providecommand{\bibinfo}[2]{#2}
\providecommand{\BIBentrySTDinterwordspacing}{\spaceskip=0pt\relax}
\providecommand{\BIBentryALTinterwordstretchfactor}{4}
\providecommand{\BIBentryALTinterwordspacing}{\spaceskip=\fontdimen2\font plus
\BIBentryALTinterwordstretchfactor\fontdimen3\font minus
  \fontdimen4\font\relax}
\providecommand{\BIBforeignlanguage}[2]{{%
\expandafter\ifx\csname l@#1\endcsname\relax
\typeout{** WARNING: IEEEtran.bst: No hyphenation pattern has been}%
\typeout{** loaded for the language `#1'. Using the pattern for}%
\typeout{** the default language instead.}%
\else
\language=\csname l@#1\endcsname
\fi
#2}}
\providecommand{\BIBdecl}{\relax}
\BIBdecl

\bibitem{Salem19}
A.~Salem, Y.~Zhang, M.~Humbert, M.~Fritz, and M.~Backes, ``{ML-Leaks: Model and
  Data Independent Membership Inference Attacks and Defenses on Machine
  Learning Models},'' in \emph{Proc. of NDSS}, San Diego, CA, USA, Feb. 2019,
  pp. 1--15.

\bibitem{Shokri17}
R.~Shokri, M.~Stronati, C.~Song, and V.~Shmatikov, ``{Membership Inference
  Attacks against Machine Learning Models},'' in \emph{Proc. of IEEE Symposium
  on Security and Privacy}, San Jose, CA, USA, May 2017, pp. 3--18.

\bibitem{Nasr19}
M.~Nasr, R.~Shokri, and A.~Houmansadr, ``{Comprehensive Privacy Analysis of
  Deep Learning: Passive and Active White-box Inference Attacks against
  Centralized and Federated Learning},'' in \emph{Proc. of IEEE Symposium on
  Security and Privacy}, San Francisco, CA, USA, May 2019, pp. 739--753.

\bibitem{Fred15}
M.~Fredrikson, S.~Jha, and T.~Ristenpart, ``{Model Inversion Attacks That
  Exploit Confidence Information and Basic Countermeasures},'' in \emph{Proc.
  of CCS}, Denver, Colorado, USA, Oct. 2015, pp. 1322--1333.

\bibitem{Yang19}
Z.~Yang, J.~Zhang, E.~Chang, and Z.~Liang, ``{Neural Network Inversion in
  Adversarial Setting via Background Knowledge Alignment},'' in \emph{Proc. of
  CCS}, London, UK, Nov. 2019, pp. 225--240.

\bibitem{Nasr18}
M.~Nasr, R.~Shokri, and A.~Houmansadr, ``{Machine Learning with Membership
  Privacy using Adversarial Regularization},'' in \emph{Proc. of CCS}, Toronto,
  ON, Canada, Oct. 2018, pp. 634--646.

\bibitem{Jia19}
J.~Jia, A.~Salem, M.~Backes, Y.~Zhang, and N.~Z. Gong, ``{MemGuard: Defending
  against Black-Box Membership Inference Attacks via Adversarial Examples},''
  in \emph{Proc. of CCS}, London, UK, Nov. 2019, pp. 259--274.

\bibitem{Papernot16}
N.~Papernot, P.~McDaniel, and I.~Goodfellow, ``{Transferability in Machine
  Learning: from Phenomena to Black-Box Attacks using Adversarial Samples},''
  in \emph{https://arxiv.org/abs/1605.07277}, 2016.

\bibitem{Yang20}
Z.~Yang, B.~Shao, B.~Xuan, E.-C. Chang, and F.~Zhang, ``{Defending Model
  Inversion and Membership Inference Attacks via Prediction Purification},'' in
  \emph{https://arxiv.org/abs/2005.03915}, 2020.

\bibitem{Abadi16}
M.~Abadi, A.~Chu, I.~Goodfellow, H.~B. McMahan, I.~Mironov, K.~Talwar, and
  L.~Zhang, ``{Deep Learning with Differential Privacy},'' in \emph{Proc. of
  CCS}, Vienna, Austria, Oct. 2016, pp. 308--318.

\bibitem{Jaya19}
B.~Jayaraman and D.~Evans, ``{Evaluating Differentially Private Machine
  Learning in Practice},'' in \emph{Proc. of USENIX Security Symposium}, Santa
  Clara, CA, USA, Aug. 2019, pp. 1895--1912.

\bibitem{Goodfellow15}
I.~Goodfellow, J.~Shlens, and C.~Szegedy, ``{Explaining and Harnessing
  Adversarial Examples},'' in \emph{Proc. of ICLR}, San Diego, CA, USA, May
  2015, pp. 1--11.

\bibitem{Papernot16SP}
N.~Papernot, P.~McDaniel, S.~Jha, M.~Fredrikson, Z.~B. Celik, and A.~Swami,
  ``{The Limitations of Deep Learning in Adversarial Settings},'' in
  \emph{Proc. of IEEE European Symposium on Security and Privacy},
  Saarbruecken, Germany, Mar. 2016, pp. 372--387.

\bibitem{Shokri15}
R.~Shokri and V.~Shmatikov, ``{Privacy-Preserving Deep Learning},'' in
  \emph{Proc. of CCS}, Denver, Colorado, US, Oct. 2015, pp. 1310--1321.

\bibitem{Gong20}
M.~Gong, Y.~Xie, K.~Pan, K.~Feng, and A.~K. Qin, ``{A Survey on Differentially
  Private Machine Learning},'' \emph{IEEE Computational Intelligence Magazine},
  vol.~15, no.~2, pp. 49--64, 2020.

\bibitem{Heikkila17}
M.~Heikkila, E.~Lagerspetz, S.~Kaski, K.~Shimizu, S.~Tarkoma, and A.~Honkela,
  ``{Differentially Private Bayesian Learning on Distributed Data},'' in
  \emph{Proc. of NIPS}, Long Beach, CA, USA, Dec. 2017, pp. 1--10.

\bibitem{Zhao19}
L.~Zhao, Q.~Wang, Q.~Zou, Y.~Zhang, and Y.~Chen, ``{Privacy-Preserving
  Collaborative Deep Learning with Unreliable Participants},'' \emph{IEEE
  Transactions on Information Forensics and Security}, vol.~15, pp. 1486--1500,
  2020.

\bibitem{Cheng18}
H.~Cheng, P.~Yu, H.~Hu, F.~Yan, S.~Li, H.~Li, and Y.~Chen, ``{LEASGD: an
  Efficient and Privacy-Preserving Decentralized Algorithm for Distributed
  Learning},'' in \emph{Proc. of NIPS Workshop on Privacy Preserving Machine
  Learning}, Montreal, Canada, Dec. 2018, pp. 1--5.

\bibitem{Jaya18}
B.~Jayaraman, L.~Wang, D.~Evans, and Q.~Gu, ``{Distributed Learning without
  Distress: Privacy-Preserving Empirical Risk Minimization},'' in \emph{Proc.
  of NIPS}, Montreal, Canada, Dec. 2018, pp. 1--12.

\bibitem{Phan19}
N.~Phan, M.~N. Vu, Y.~Liu, R.~Jin, D.~Dou, X.~Wu, and M.~T. Thai,
  ``{Heterogeneous Gaussian Mechanism: Preserving Differential Privacy in Deep
  Learning with Provable Robustness},'' in \emph{Proc. of IJCAI}, Macao, China,
  Aug. 2019, pp. 4753--4759.

\bibitem{Papernot17}
N.~Papernot, M.~Abadi, U.~Erlingsson, I.~Goodfellow, and K.~Talwar,
  ``{Semi-supervised Knowledge Transfer for Deep Learning from Private Training
  Data},'' in \emph{Proc. of ICLR}, Toulon, France, Apr. 2017, pp. 1--16.

\bibitem{Papernot18}
N.~Papernot, S.~Song, I.~Mironov, A.~Raghunathan, K.~Talwar, and U.~Erlingsson,
  ``{Scalable Private Learning with PATE},'' in \emph{Proc. of ICLR},
  Vancouver, BC, Canada, May 2018, pp. 1--16.

\bibitem{Jordon19}
J.~Jordon, J.~Yoon, and M.~van~der Schaar, ``{PATE-GAN: Generating Synthetic
  Data with Differential Privacy Guarantees},'' in \emph{Proc. of ICLR}, New
  Orleans, Louisiana, US, May 2019, pp. 1--12.

\bibitem{Kim21}
M.~Kim, O.~Gunlu, and R.~F. Schaefer, ``{Federated Learning with Local
  Differential Privacy: Trade-offs between Privacy, Utility and
  Communication},'' in \emph{to appear in Proc. of IEEE International
  Conference on Acoustics, Speech, and Signal Processing}, Toronto, Ontario,
  Canada, Jun. 2021.

\bibitem{Dwork14}
C.~Dwork and A.~Roth, ``{The Algorithmic Foundations of Differential
  Privacy},'' \emph{Foundations and Trends in Theoretical Computer Science},
  vol.~9, no. 3-4, pp. 211--407, 2014.

\bibitem{Zhu20}
T.~Zhu, D.~Ye, W.~Wang, W.~Zhou, and P.~S. Yu, ``{More Than Privacy: Applying
  Differential Privacy in Key Areas of Artificial Intelligence},'' \emph{IEEE
  Transactions on Knowledge and Data Engineering}, p. DOI:
  10.1109/TKDE.2020.3014246, 2020.

\bibitem{Ye19}
D.~Ye, T.~Zhu, W.~Zhou, and P.~S. Yu, ``{Differentially Private Malicious Agent
  Avoidance in Multiagent Advising Learning},'' \emph{IEEE Transactions on
  Cybernetics}, vol.~50, no.~10, pp. 4214--4227, 2020.

\bibitem{Ye20b}
D.~Ye, T.~Zhu, Z.~Cheng, W.~Zhou, and P.~S. Yu, ``{Differential advising in
  multiagent reinforcement learning},'' \emph{IEEE Transactions on
  Cybernetics}, p. DOI: 10.1109/TCYB.2020.3034424, 2020.

\bibitem{Ye20}
D.~Ye, T.~Zhu, S.~Shen, W.~Zhou, and P.~S. Yu, ``{Differentially Private
  Multi-Agent Planning for Logistic-like Problems},'' \emph{IEEE Transactions
  on Dependable and Secure Computing}, p. DOI: 10.1109/TDSC.2020.3017497, 2020.

\bibitem{Ye21}
D.~Ye, T.~Zhu, S.~Shen, and W.~Zhou, ``{A Differentially Private Game Theoretic
  Approach for Deceiving Cyber Adversaries},'' \emph{IEEE Transactions on
  Information Forensics and Security}, vol.~16, pp. 569--584, 2021.

\bibitem{Dwork18}
C.~Dwork and V.~Feldman, ``{Privacy-preserving Prediction},'' in \emph{Proc. of
  31st annual Conference on Learning Theory (COLT)}, Stockholm, Sweden, Jul.
  2018, pp. 1--10.

\bibitem{Zhu17}
T.~Zhu, G.~Li, W.~Zhou, and P.~S. Yu, ``{Differentially private data publishing
  and analysis: A survey},'' \emph{IEEE Transactions on Knowledge and Data
  Engineering}, vol.~29, no.~8, pp. 1619--1638, 2017.

\bibitem{Ganesh20}
A.~Ganesh and K.~Talwar, ``{Faster Differentially Private Samplers via Renyi
  Divergence Analysis of Discretized Langevin MCMC},'' in \emph{Proc. of NIPS},
  Vancouver, Canada, Dec. 2020, pp. 1--12.

\bibitem{LeCun98}
Y.~LeCun, ``The mnist database of handwritten digits,'' in
  \emph{http://yann.lecun.com/exdb/mnist/}, 1998.

\bibitem{Fashion}
Fashion-MNIST, ``{An MNIST-like dataset of 70,000 28x28 labeled fashion
  images},'' in \emph{https://www.kaggle.com/zalando-research/fashionmnist}.

\bibitem{Krizhevsky14}
A.~Krizhevsky, V.~Nair, and G.~Hinton, ``The cifar-10 dataset,'' in
  \emph{http://www.cs.toronto.edu/kriz/cifar.html}, 2014.

\bibitem{Bozkir20}
E.~Bozkir, O.~Gunl, W.~Fuhl, R.~F. Schaefer, and E.Kasneci, ``{Differential
  Privacy for Eye Tracking with Temporal Correlations},'' in
  \emph{https://arxiv.org/pdf/2002.08972.pdf}, 2020.

\bibitem{Wang04}
Z.~Wang, A.~C. Bovik, H.~R. Sheikh, and E.~P. Simoncelli, ``{Image quality
  assessment: from error visibility to structural similarity},'' \emph{IEEE
  Transactions on Image Processing}, vol.~13, no.~4, pp. 600--612, 2004.

\bibitem{Choo20}
C.~A.~C. Choo, F.~Tramer, N.~Carlini, and N.~Papernot, ``{Label-Only Membership
  Inference Attacks},'' in \emph{https://arxiv.org/abs/2007.14321}, 2020.

\bibitem{Li21}
Z.~Li and Y.~Zhang, ``{Membership Leakage in Label-Only Exposures},'' in
  \emph{Proc. of CCS}, 2021.

\end{thebibliography}

%\vfill

\end{document}